\documentclass[11pt,draftcls,onecolumn]{IEEEtran}
\usepackage{cite}

\IEEEoverridecommandlockouts                              % This command is only
\usepackage{subfigure}
\usepackage{graphicx}
\usepackage{latexsym}
\usepackage{color}
\usepackage{amsmath}
\usepackage{amsthm}
\usepackage{amstext,amssymb, amsfonts}
\usepackage{chngcntr}
\usepackage{array}
\def\mindex#1{\index{#1}}

\def\sq{\hbox{\rlap{$\sqcap$}$\sqcup$}}
\def\qed{\ifmmode\sq\else{\unskip\nobreak\hfil
\penalty50\hskip1em\null\nobreak\hfil\sq
\parfillskip=0pt\finalhyphendemerits=0\endgraf}\fi\medskip}

\long\def\defbox#1{\framebox[.9\hsize][c]{\parbox{.85\hsize}{%
\parindent=0pt
\baselineskip=12pt plus .1pt      % STYLE
\parskip=6pt plus 1.5pt minus 1pt % CHANGES
 #1}}}

\long\def\beginbox#1\endbox{\subsection*{}%
\hbox{\hspace{.05\hsize}\defbox{\medskip#1\bigskip}}%
\subsection*{}}

\def\endbox{}

\newsavebox{\junk}
\savebox{\junk}[1.6mm]{\hbox{$|\!|\!|$}}

\def\liminf{\mathop{\rm lim\ inf}}
\def\argmin{\mathop{\rm arg\, min}}

\def\argmax{\mathop{\rm arg\, max}}

\def\ystate{{\sf Y}}

\def\zstate{{\sf Z}}

\def\bfZ{{\bf Z}}

\def\bfmath#1{{\mathchoice{\mbox{\boldmath$#1$}}%
{\mbox{\boldmath$#1$}}%
{\mbox{\boldmath$\scriptstyle#1$}}%
{\mbox{\boldmath$\scriptscriptstyle#1$}}}}

\def\bfmY{\bfmath{Y}}

\def\bfmhhaY{\bfmath{\hhaY}} %\widehat{\widehat{Y}}}}
\def\bfmhhaY{\hbox to 0pt{$\widehat{\bfmY}$\hss}\widehat{\phantom{\raise 1.25pt\hbox{$\bfmY$}}}}

\def\til={{\widetilde =}}

\def\clG{{\cal G}}
\def\clH{{\cal H}}

\def\clI{{\cal I}}
\def\clK{{\cal K}}

\def\clM{{\cal M}}
\def\clN{{\cal N}}

\def\clP{{\cal P}}

\def\clS{{\cal S}}

\def\clU{{\cal U}}
\def\clV{{\cal V}}

 \def\FRAC#1#2#3{\genfrac{}{}{}{#1}{#2}{#3}}

\def\ddtp{{\mathchoice{\FRAC{1}{d^{\hbox to 2pt{\rm\tiny +\hss}}}{dt}}%
{\FRAC{1}{d^{\hbox to 2pt{\rm\tiny +\hss}}}{dt}}%
{\FRAC{3}{d^{\hbox to 2pt{\rm\tiny +\hss}}}{dt}}%
{\FRAC{3}{d^{\hbox to 2pt{\rm\tiny +\hss}}}{dt}}}}

\def\half{{\mathchoice{\FRAC{1}{1}{2}}%
{\FRAC{1}{1}{2}}%
{\FRAC{3}{1}{2}}%
{\FRAC{3}{1}{2}}}}

\def\Prob{{\sf P}}

\def\average#1,#2,{{1\over #2} \sum_{#1}^{#2}}

\def\eye(#1){{\bf(#1)}\quad}

\newtheorem{theorem}{Theorem}[section]

\newtheorem{lemma}[theorem]{Lemma}

\def\eq#1/{(\ref{e:#1})}

\newcommand{\beqn}[1]{\notes{#1}%
\begin{eqnarray} \elabel{#1}}

\newcommand{\eeqn}{\end{eqnarray} }

\newcommand{\beq}[1]{\notes{#1}%
\begin{equation}\elabel{#1}}

\newcommand{\eeq}{\end{equation}}

\def\bdes{\begin{description}}
\def\edes{\end{description}}

\newcounter{rmnum}

\newcounter{anum}

{\end{list}}

\def\ass(#1:#2){(#1\ref{#1:#2})}

\def\ritem#1{
\item[{\sf \ass(\current_model:#1)}]
}

\newenvironment{recall-ass}[1]{%
\begin{description}
\def\current_model{#1}}{
\end{description}
}

\newcommand{\bd}{\begin{description}}
\newcommand{\ed}{\end{description}}
\newcommand{\bt}{\begin{theorem}}
\newcommand{\et}{\end{theorem}}
\newcommand{\ba}{\begin{array}{rcl}}
\newcommand{\ea}{\end{array}}

\usepackage{booktabs}

\usepackage{multirow}

\usepackage{psfrag}

\usepackage{mathtools}

\makeatletter
\DeclareRobustCommand\widecheck[1]{{\mathpalette\@widecheck{#1}}}
\def\@widecheck#1#2{%
    \setbox\z@\hbox{\m@th$#1#2$}%
    \setbox\tw@\hbox{\m@th$#1%
       \widehat{%
          \vrule\@width\z@\@height\ht\z@
          \vrule\@height\z@\@width\wd\z@}$}%
    \dp\tw@-\ht\z@
    \@tempdima\ht\z@ \advance\@tempdima2\ht\tw@ \divide\@tempdima\thr@@
    \setbox\tw@\hbox{%
       \raise\@tempdima\hbox{\scalebox{1}[-1]{\lower\@tempdima\box
\tw@}}}%
    {\ooalign{\box\tw@ \cr \box\z@}}}
\makeatother

\def\bfZ{\mathbf{Z}}
\def\pz{\clP(\zstate)}
\def\wtilde{\widetilde}
\def\what{\widehat}
\def\ystate{\mathsf{Y}}
\def\py{\clP(\ystate)}
\def\Sym{\mathsf{Sym}}
\def\err{\mathsf{err}}
\def\rej{\mathsf{rej}}
\def\uc{\mathsf{uc}}
\def\match{\mathsf{M}}
\newcounter{prb}
\newenvironment{problem}{\begin{list}{{ \bf (P\arabic{prb}) \ }}{\usecounter{prb}
\setlength{\leftmargin}{14pt}
\setlength{\rightmargin}{12pt}
\setlength{\itemindent}{-25 pt}
\setlength{\labelsep}{-3 pt}
\setlength{\itemsep}{5 pt}
}}{\end{list}}

\def\Prb#1{\textbf{(P\ref{prb:#1})}}
\newtheorem{exmp}{Example}[section]

\begin{document}
%%%%%%%%%%%%%%%%%%%%%%%%%%%%%%%%TITLE%%%%%%%%%%%%%%%%%%%%%%%%%%%%%%%%%%%%%%%%
%\title{Asymptotically Optimal Matching of Unlabeled Observation Sequences}
\title{Asymptotically Optimal Matching of Multiple Sequences to Source Distributions and Training Sequences}
\author{
%\IEEEauthorblockN{Jayakrishnan Unnikrishnan\IEEEauthorrefmark{1}, Farid M. Naini\IEEEauthorrefmark{1,2}}
%    \IEEEauthorblockA{\IEEEauthorrefmark{1}School of Computer and Communication Sciences\\
%Ecole Polytechnique F\'{e}d\'{e}rale de Lausanne (EPFL),
%CH-1015 Lausanne, Switzerland}
%    \IEEEauthorblockA{\IEEEauthorrefmark{2}Institution2
%    \\\{2, 3\}@def.com}
\IEEEauthorblockN{Jayakrishnan Unnikrishnan\\}
\IEEEauthorblockA{Audiovisual Communications Laboratory, School of Computer and Communication Sciences\\
Ecole Polytechnique F\'{e}d\'{e}rale de Lausanne (EPFL),
CH-1015 Lausanne, Switzerland\\
Email:  jay.unnikrishnan@epfl.ch}
\thanks{Portions of this paper were presented in part at the 51st Annual Allerton Conference on Communication, Control, and Computing, Monticello, Illinois, October 2013 \cite{unnnai13}.}
}
\maketitle
\thispagestyle{empty}
\pagestyle{empty}

%%%%%%%%%%%%%%%%%%%%%%%%%%%%%%%%%%%%%%%%%%%%%%%%%%%%%%%%%%%%%%%%%
%%%%%%%%%%%%%%%%%%%%%%%%%%%%%%%%%%%%%%%%%%%%%%%%%%%%%%%%%%%%%%%%%
%%%%%%%%%%%%%%%%%%%%%%%%%%%%%%%%%%%%%%%%%%%%%%%%%%%%%%%%%%%%%%%%%

\begin{abstract}
Consider a finite set of sources, each producing i.i.d. observations that follow a unique probability distribution on a finite alphabet.
We study the problem of matching a finite set of observed sequences to the set of sources under the constraint that the observed sequences are produced by distinct sources.
In general, the number of sequences $N$ may be different from the number of sources $M$, and only some $K \leq \min\{M,N\}$ of the observed sequences may be produced by a source from the set of sources of interest.
We consider two versions of the problem -- one in which the probability laws of the sources are known, and another in which the probability laws of the sources are unspecified but one training sequence from each of the sources is available.
We show that both these problems can be solved using a sequence of tests that are allowed to produce ``no-match'' decisions.
The tests ensure exponential decay of the probabilities of incorrect matching as the sequence lengths increase, and minimize the ``no-match'' decisions.
Both tests can be implemented using variants of the minimum weight matching algorithm applied to a weighted bipartite graph.
We also compare the performances obtained by using these tests with those obtained by using tests that do not take into account the constraint that the sequences are produced by distinct sources.
For the version of the problem in which the probability laws of the sources are known, we compute the rejection exponents and error exponents of the tests and show that tests that make use of the constraint have better exponents than tests that do not make use of this information.
%Consider a set of $M$ sources, each producing i.i.d. observations that follow a unique probability distribution on a finite alphabet.
%We study the problem of matching a set of $N$ observed sequences to the set of $M$ sources, assuming that the observed sequences are produced by distinct sources.
%In general, $N$ may be different from $M$, and only some $K \leq \min\{M,N\}$ observation sequences may be produced by a source from the set considered.
%We consider two versions of the problem -- one in which the probability laws of the sources are known, and another in which the probability laws of the sources are unspecified but one training sequence from each of the sources is available.
%We show that both these problems can be solved using likelihood ratio tests that are allowed to produce ``no-match'' decisions.
%The tests ensure exponential decay of the probabilities of incorrect matching as the sequence lengths increase, and minimize the ``no-match'' decisions.
%Both tests can be implemented using variants of the minimum weight matching algorithm applied to a weighted bipartite graph.
\end{abstract}

%%%%%%%%%%%%%%%%%%%%%%%%%%%%%%%%%%%%%%%%%%%%%%%%%%%%%%%%%%%%%%%%%
%%%%%%%%%%%%%%%%%%%%%%%%%%%%%%%%%%%%%%%%%%%%%%%%%%%%%%%%%%%%%%%%%
%%%%%%%%%%%%%%%%%%%%%%%%%%%%%%%%%%%%%%%%%%%%%%%%%%%%%%%%%%%%%%%%%

\section{Introduction}\label{sec: intro}

Classical multi-hypothesis testing \cite{leh05} addresses the following problem: Given probability distributions of $M$ sources and one observation sequence (or string), decide which of the $M$ sources produced the sequence.
Classical statistical classification \cite{hastibfri09} also addresses the same problem with the only difference that the probability distributions of the sources are not known exactly, but instead, have to be estimated from training sequences produced by the sources.
Figure~\ref{fig:classprobs} illustrates these classical problems.
In this paper we study a generalization of these problems which is relevant in applications like de-anonymization of anonymized data \cite{unnnai13}.
Instead of one observation sequence, suppose that you are given $N$ observation sequences, subject to the constraint that each sequence is produced by a distinct source.
We consider the task of matching the sequences to the correct sources that produced them, as illustrated in Figure~\ref{fig:matchprobsnew}.
Focusing on finite alphabet sources, we study these matching problems as composite hypothesis testing problems.
We refer to the first problem, in which the distributions of the sources are known, as the \emph{matching problem with known sources}, and the second problem in which only training sequences under the sources are given, as the \emph{matching problem with unknown sources}.
We obtain solutions to both these problems that are asymptotically optimal in error probability as the length of the sequences increases to infinity.

The main difference between these problems and the standard multi-hypothesis testing and classification problems is the constraint that the observation sequences are produced by distinct sources.
It is clear that in the absence of such a constraint, these problems are just repeated versions of the standard problems.
The constraint adds more structure to the solution and leads to an improvement in classification accuracy.
We use large deviations analysis to quantify the improvement in performance in terms of the asymptotic rate of decay of the error probabilities and the probabilities of rejection of the optimal tests with and without the constraints.
We obtain asymptotically optimal solutions to these matching problems using a generalization of the approach of Gutman \cite{gut89}, who solved the classical statistical classification problem.

\begin{figure*}%[ht]
\centering
\subfigure[Multihypothesis testing: Match observed string to correct source probability mass function (PMF)]{
\psfrag{s1}{}
\psfrag{s2}{Source PMFs}
\psfrag{p1}{$\mu_1$}
\psfrag{p2}{$\mu_2$}
\psfrag{p3}{$\mu_3$}
\psfrag{p4}{$\mu_M$}
\psfrag{o1}{}
\psfrag{vd}{\Huge $\vdots$}
\psfrag{o2}{Observed string}
\psfrag{y1}{$y_1$}
\includegraphics[width=2.5in]
%{Figures/fig1a.eps}
{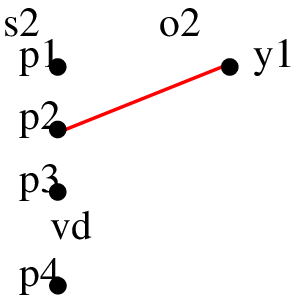}
\label{fig:MHT}
}
\subfigure[Statistical classification: Match observed string to the training string from the correct source]{
\psfrag{s2}{Source strings}
\psfrag{o2}{Observed string}
\psfrag{p1}{$x_1$}
\psfrag{p2}{$x_2$}
\psfrag{p3}{$x_3$}
\psfrag{p4}{$x_M$}
\psfrag{vd}{\Huge $\vdots$}
\psfrag{y1}{$y_1$}
\includegraphics[width=2.5in]
%{Figures/fig1b.eps}
{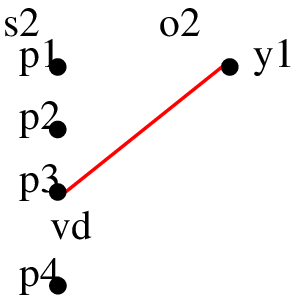}
\label{fig:SC}
}
\caption[]{Illustration of the matching tasks in classical multihypothesis testing and statistical classification}
\label{fig:classprobs}
\end{figure*}

\begin{figure*}%[ht]
\centering
\subfigure[Matching problem~\Prb{known}: Match observed strings to correct source PMFs]{
\psfrag{s1}{}
\psfrag{s2}{Source PMFs ($\clM$)}
\psfrag{p1}{$\mu_1$}
\psfrag{p2}{$\mu_2$}
\psfrag{p3}{$\mu_3$}
\psfrag{p4}{$\mu_M$}
\psfrag{o1}{}
\psfrag{vd}{\Huge $\vdots$}
\psfrag{o2}{Observed strings ($\clS$)}
\psfrag{y1}{$y_1$}
\psfrag{y2}{$y_2$}
\psfrag{y3}{$y_3$}
\psfrag{y4}{$y_N$}
\includegraphics[width=2.5in]
%{Figures/fig1a.eps}
{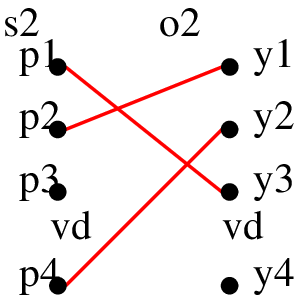}
\label{fig:P1}
}
\subfigure[Matching problem~\Prb{unknown}: Match observed strings to the training strings from the correct sources]{
\psfrag{s2}{Source strings ($\clS_1$)}
\psfrag{o2}{Observed strings ($\clS_2$)}
\psfrag{p1}{$x_1$}
\psfrag{p2}{$x_2$}
\psfrag{p3}{$x_3$}
\psfrag{p4}{$x_M$}
\psfrag{vd}{\Huge $\vdots$}
\psfrag{y1}{$y_1$}
\psfrag{y2}{$y_2$}
\psfrag{y3}{$y_3$}
\psfrag{y4}{$y_N$}
\includegraphics[width=2.5in]
%{Figures/fig1b.eps}
{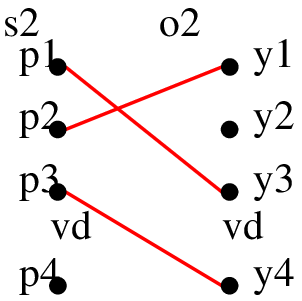}
\label{fig:P2}
}
\caption[]{Illustration of the matching tasks in the problems studied in this paper}
\label{fig:matchprobsnew}
\end{figure*}

Our primary motivation for studying these problems comes from studies on privacy of anonymized databases.
In recent years, many datasets containing information about individuals have been released into public domain in order to provide open access to statistics or to facilitate data mining research.
Often these databases are \emph{anonymized} by suppressing identifiers that reveal the identities of the users, like names or social security numbers.
Nevertheless, recent research (see, e.g., \cite{narshm08,monhidverblo13}) has revealed that the privacy offered by such anonymized databases may be compromised if an adversary correlates the revealed information with publicly available databases.
%For instance, in \cite{narshm08} it was shown that anonymous movie ratings released during the Netflix Prize context could be de-anonymized using public user reviews from the Internet Movie Database (IMDB), and more recently,
In our recent work \cite{unnnai13}, we studied the privacy of anonymized user statistics in the presence of auxiliary information.
We showed that anonymized statistical information about a set of users can be easily de-anonymized by an adversary who has access to independent auxiliary observations about the users.
The task of the adversary is to match the auxiliary information to the anonymized statistics, which is exactly the problem that is studied in the current paper.
Another related application of the matching problem is in matching statistical profiles of users obtained from two different sources.
For example, one could obtain location statistics of a set of users either from connections to WiFi access points, or from connections to mobile towers.
It is interesting to try to match users across these two datasets, using the statistics of their location patterns.
Alternatively, the user statistics  could be the frequency distributions of words used by users on two different blog websites.
The matching task then is to identify users who have accounts on both websites.
Matching users across two different datasets increases the net information available about the users which in turn can be used to improve accuracy of targeted services.

Asymptotically optimal hypothesis testing has a long history in literature (see e.g., \cite{hoe65a, bla74a, tus77, unnhuameysurvee11}).
However, hypothesis testing of multiple sequences under the constraint that each sequence is produced by a distinct source, has been studied only rarely.
The prior knowledge of the constraint on the sequences is expected to improve the accuracy of the hypothesis test.
However, the task of identifying the optimal solution is now much more complicated as there are a combinatorial number of hypotheses.
It is not immediately clear what is the best strategy to adopt.
A naive strategy is to try to classify each sequence individually; but that is not expected to yield high accuracy as the constraints are not intelligently exploited.
In \cite[Ch. 10]{ahlweg87} the matching problem with known distributions was studied for the special case of $M=N$, where the analysis was performed by reducing the problem to a multi-hypothesis testing problem.
The same problem was solved in \cite{ahlhar06} for $M=N=2$ under a different optimality criterion from that used in this paper.
In the first part of this paper we study this problem under a different optimality criterion, and identify an optimal solution for general $M$ and $N$.
We also provide a quantitative comparison of the performance obtained with our optimal solution, with that of a test that ignores the constraint on the sequences.
These results quantify the improvement in performance that can be obtained by exploiting the constraint on the sequences.
In the second part of this paper we study the problem of matching one set of sequences to another set of sequences.
The approach we adopt in most of this paper is a generalization of that adopted by Gutman in \cite{gut89}, who solved the matching problem with unknown source distributions for $N=1$.
Gutman showed that if a ``no-match'' decision is allowed, it is possible to guarantee exponential decay of all misclassification probabilities at a desired rate.
In the second part of this paper we simplify the structure of Gutman's solution and show that his method can be generalized to solve the matching problem with unknown source distributions for general $N$.
We also demonstrate that although there are a combinatorial number of hypotheses for these problems, simple polynomial-time algorithms can be used to identify the optimal solutions.

The rest of the paper is organized as follows.
After introducing our notation, we state the problems in mathematical form in Section~\ref{sec:problemstatement}.
We present our solution to the generalization of the hypothesis testing problem in Section~\ref{sec:optknown} and our solution to the generalization of the statistical classification problem in Section~\ref{sec:optunknown}.
In addition to identifying the optimal solutions, we also compare the performances of these solutions with those of solutions that do not explicitly take into account the constraint on the distinctness of the sources that produced the sequences.
%, and  experimentally evaluate it in Section~\ref{sec:experiments}.
We discuss practical aspects of implementing the test and conclude in Section~\ref{sec:conc}.
For ease of reading, we relegate proofs of all results to the appendix.
\medskip

\noindent \textit{Notation: }
For a finite alphabet $\zstate$, we use $\pz$ to denote the set of all probability distributions defined on $\zstate$.
We interchangeably use the words sequence and string to refer to an ordered list of elements from $\zstate$.
For any string $s \in \zstate^n$, we use $\Gamma_{s} \in \pz$ to denote the empirical distribution of the string defined as
\[
\Gamma_{s}(z) = \frac{1}{n}\sum_{i=1}^n \clI\{s_i = z\}, z \in \zstate.
\]
%Further we use $T_s$ to denote the \emph{type class} of $s$, i.e., the set of all strings of length $n$ with the same empirical distribution as $s$.
For $\mu \in \pz$ and $z \in \zstate$ we use $\mu(z)$ to denote the probability mass at $z$ under $\mu$.
For a string $s \in \zstate^n$, we use $\mu(s)$ to denote the probability of observing $s$ at the output of a source that generates $n$ observations i.i.d. according to law $\mu$.
We use $H(\mu)$ to denote the Shannon entropy
\[
H(\mu) = \sum_{z \in \zstate} -\mu(z) \log \mu(z).
\]
For $\nu, \mu \in \pz$ we use
\[
D(\nu \| \mu) = \sum_{z \in \zstate} \nu(z) \log \frac{\nu(z)}{\mu(z)},
\]
to denote the Kullback-Leibler divergence between probability distributions $\nu$ and $\mu$.
Throughout the paper we use $\log$ to refer to logarithm to the base $2$.
We use $[N] := \{1,2,\ldots,N\}$ and $\Sym([N])$ to denote the set of all permutations on $[N]$, i.e., if $\sigma \in \Sym([N])$ then $\sigma$ is a one-to-one mapping from $[N]$ onto itself.

%The publicly released anonymized information comprises the empirical distributions of the data of each user collected over one day.
%We consider an adversary who has access to an additional database with similar non-anonymized empirical distribution of the users' data.
%The objective of the adversary is to match the data between the two days.
%A correct matching would reveal the identities of the users in the anonymized dataset.

%%%%%%%%%%%%%%%%%%%%%%%%%%%%%%%%%%%%%%%%%%%%%%%%%%%%%%%%%%%%%%%%%
%%%%%%%%%%%%%%%%%%%%%%%%%%%%%%%%%%%%%%%%%%%%%%%%%%%%%%%%%%%%%%%%%
%%%%%%%%%%%%%%%%%%%%%%%%%%%%%%%%%%%%%%%%%%%%%%%%%%%%%%%%%%%%%%%%%

\section{Problem Statement}\label{sec:problemstatement}
Consider a set of independent sources each producing i.i.d. data according to distinct but unknown probability distributions on a finite alphabet $\zstate$.
Let $\clU \subset \pz$ denote the set of probability distributions followed by these sources.
%Associated with each distribution in $\clU$ is a source each producing i.i.d. data according to their respective distribution.
Let $\clM \subseteq \clU$ and $\clN \subseteq \clU$ be such that $\clM \cap \clN = \clK$.
Let $|\clM| = M$, $|\clN| = N$ and $|\clK| = K$.
We are concerned with the following two problems.
\begin{problem}
\item \label{prb:known} \emph{[Known sources]} Let $\clM = \{\mu_1,\mu_2,\ldots,\mu_M\}$. Suppose $\clM$, $N$ and $K$ are known but $\clN$ and $\clK$ are not. Further, suppose a set $\clS= \{y_1, y_2, \ldots, y_N\}$ of unlabeled sequences of length $n$ each generated independently under a distinct distribution in $\clN$ is given. Identify the $K$ sequences in $\clS$ that were generated under distributions in $\clM$, and match each of these sequences to the correct distribution in $\clM$ that generated it.

\item \label{prb:unknown} \emph{[Unknown sources]} Suppose the distributions are unknown, but $M$, $N$ and $K$ are known. Given a set $\clS_1= \{x_1, x_2, \ldots, x_M\}$ of unlabeled sequences of length $n$ each generated under a distinct distribution in $\clM$, and a set $\clS_2= \{y_1, y_2, \ldots, y_N\}$ of unlabeled sequences of length $n$ each generated under a distinct distribution in $\clN$, identify the $K$ sequences in $\clS_1$ that were generated by distributions in $\clK$ and match each of them to the sequence in $\clS_2$ that was generated under the same distribution. The information in $\clS_1$ and $\clS_2$ are assumed to be independent of each other.
\end{problem}

As mentioned earlier, such problems arise in the fields of de-anonymization of databases, and of identification of users from the statistics of their data.
For example, $\clS_1$ and $\clS_2$ in problem~\Prb{unknown} could be two anonymized databases of data belonging to known sets of users.
It may be known that the two sets of users are identical, in which case $\clM = \clN = \clK$, or it may be that the second set of users is a subset of the first set, in which case $\clM \supset \clN = \clK$.
In some other cases, the sets $\clM$ and $\clN$ might not be subsets but the statistician may have an estimate for the number $K$ of common users in the two sets, i.e., the size of $\clK$.
Problem~\Prb{known} arises when the statistical behavior of the data belonging to the first set of users is known accurately.
As stated above, for simplifying the analysis, in both these problems we have assumed that the sample size of all sequences are equal, and that the alphabet $\zstate$ is a finite set.
In Section~\ref{sec:conc} we discuss how the analyses and results can be generalized to the setting in which the sequence lengths are not equal, or when the alphabet is continuous.

Both problems~\Prb{known} and~\Prb{unknown} can be visualized as variants of the following problem.
Let $\clV_1$, $\clV_2$ be two sets of objects with $|\clV_1| = M$ and $|\clV_2|=N$.
Consider a complete bipartite graph $\clG$ \cite{west2001introduction} with vertices $\clV = \clV_1 \cup \clV_2$ such that every vertex in $\clV_1$ is connected to every vertex in $\clV_2$ by an edge, as illustrated in Figure~\ref{fig:compbip}.
The objective is to identify a matching\footnote{A matching in a graph is a set of edges such that no two edges in the set share a common vertex.} of cardinality $K$ in the graph $\clG$ that satisfies some conditions.
In problem~\Prb{known}, the set $\clV_1 = \clM$ and $\clV_2 = \clS$, and in problem~\Prb{unknown} the set $\clV_1 = \clS_1$ and $\clV_2 = \clS_2$.
Illustrations of these matching problems are shown in Figure~\ref{fig:matchprobsnew}.
More precisely, these problems are multi-hypothesis testing problems, where each hypothesis corresponds to a potential matching of  cardinality $K$ on the graph $\clG$.
Thus, for each problem, there are a total of $J={M \choose K} {N \choose K } K!$ different hypotheses.
We let $\clH_1, \clH_2,\ldots, \clH_J$ denote an enumeration of the hypotheses for each problem.
We use the same notation for the hypotheses in both problems; it should always be clear from context what is intended.
It is to be noted that in both problems the hypotheses are composite.
In problem~\Prb{known}, the sequences in $\clS$ that do not follow a distribution in $\clK$ are allowed to have any distribution.
In problem~\Prb{unknown} the probability distributions of each source could lie anywhere in $\pz$.

\begin{figure}[h]
\centering
\psfrag{s2}{$\clV_1$}
\psfrag{o2}{$\clV_2$}

\psfrag{p1}{$1$}
\psfrag{p2}{$2$}
\psfrag{p3}{$3$}
\psfrag{p4}{$M$}
\psfrag{vd}{\Huge $\vdots$}
\psfrag{y1}{$1$}
\psfrag{y2}{$2$}
\psfrag{y3}{$3$}
\psfrag{y4}{$N$}
\includegraphics[width=.4\columnwidth]{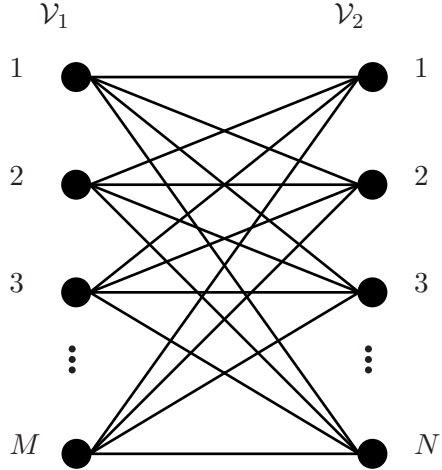}
\caption{A complete bipartite graph. Every vertex in $\clV_1$ is connected by an edge to every vertex in $\clV_2$.} \label{fig:compbip}
\end{figure}

We seek decision rules for these problems that admit exponential decay of error probabilities as a function of $n$ under each hypothesis.
For this purpose, for each problem, we allow a no-match decision, i.e., rejection of all $J$ hypotheses.
Thus a decision rule for problem~\Prb{known} is given by a partition $\Omega = (\Omega_1, \Omega_2, \ldots, \Omega_J, \Omega_R)$ of $\mathbf{Z}^1 = (\zstate^n)^N$ the space of vectors of the form $y_1, y_2, \ldots, y_N$, into $(J + 1)$ disjoint cells $\Omega_1, \Omega_2, \ldots, \Omega_J, \Omega_R$, where $\Omega_\ell$ is the acceptance region for hypothesis $\clH_\ell$ for $\ell \in [J]$, and $\Omega_R = \bfZ^1 -\cup_{\ell=1}^J \Omega_\ell$ is the rejection zone.
Similarly, a decision rule for problem~\Prb{unknown} is given by a partition $\Omega = (\Omega_1, \Omega_2, \ldots, \Omega_J, \Omega_R)$ of $\mathbf{Z} = (\zstate^n)^M \times (\zstate^n)^N$ the space of vectors of the form $x_1, x_2, \ldots, x_M, y_1, y_2, \ldots, y_N$, into $(J + 1)$ disjoint cells $\Omega_1, \Omega_2, \ldots, \Omega_J, \Omega_R$, where $\Omega_\ell$ is the acceptance region for hypothesis $\clH_\ell$, and $\Omega_R = \bfZ -\cup_{\ell=1}^J \Omega_\ell$ is the rejection zone.
In both these problems, we consider an error event $\err$ under hypothesis $\clH_\ell$ to denote a decision in favor of a wrong hypothesis $\clH_k$ where $k \neq \ell$.
We denote a decision in favor of rejection by $\rej$.
Note that a decision in favor of rejection does not correspond to an error event under any hypothesis.
Thus, using the notation $\underline x = (x_1, x_2, \ldots, x_M)$ and $\underline y = (y_1, y_2, \ldots, y_N)$, the probability of error of the decision rule $\Omega$ under hypothesis $\clH_\ell$ is given by
\begin{equation}
P_\Omega(\err/\clH_\ell) = \Prob_{\clH_\ell}\left\{\underline y \in \displaystyle\bigcup_{\substack{k=1\\k \neq \ell}}^J\Omega_k \right\}
\end{equation}
for problem~\Prb{known}, and by
\begin{equation}
%P_\Omega(\err/\clH_\ell) = \Prob_{\clH_\ell}\{(\underline x, \underline y) \in \cup_{j \in \{1, \ldots, M\} \setminus \{\ell\}} \Omega_j \}
P_\Omega(\err/\clH_\ell) = \Prob_{\clH_\ell}\left\{(\underline x, \underline y) \in \displaystyle\bigcup_{\substack{k=1\\k \neq \ell}}^J\Omega_k \right\}
\end{equation}
for problem~\Prb{unknown}.
%Here $\underline x = (x_1, x_2, \ldots, x_M)$, $\underline y = (y_1, y_2, \ldots, y_N)$, and
Here $\Prob_{\clH_\ell}$ indicates the probability measure under hypothesis $\clH_\ell$.
For both problems, we consider a generalized Neyman-Pearson criterion wherein we seek to ensure that all error probabilities decay exponentially in $n$ with some predetermined slope $\lambda$, and simultaneously minimize the rejection probability subject to these constraints. Specifically, we seek optimal decision rules $\Omega$ such that $\forall\clU \subset \pz$
\begin{equation}
\liminf_{n \to \infty} -\frac{1}{n} \log P_\Omega(\err/\clH_\ell) \geq  \lambda, \ell \in [J], \label{eqn:criterion}
\end{equation}
and $\Omega_R$ is minimal.
The quantity on the left hand side of \eqref{eqn:criterion} is called the error exponent under hypothesis $\clH_\ell$.
The optimality criterion for Problem~\Prb{unknown} is defined analogously.
This approach for identifying an optimal test with rejection was introduced by Gutman in \cite{gut89} when he solved Problem~\Prb{unknown} for $N=K=1$.

We define rejection probabilities and rejection exponents analogously to error probability and error exponents.
The probability of rejection of the decision rule $\Omega$ under hypothesis $\clH_\ell$ is given by
\begin{equation}
P_\Omega(\rej/\clH_\ell) = \Prob_{\clH_\ell}\left\{\underline y \in \Omega_R \right\} \label{eqn:rejprob}
\end{equation}
for both problem~\Prb{known} and by
\begin{equation}
P_\Omega(\rej/\clH_\ell) = \Prob_{\clH_\ell}\left\{(\underline x, \underline y) \in \Omega_R \right\} \label{eqn:rejprob2}
\end{equation}
for problem~\Prb{unknown}. %and by
%\begin{equation}
%%P_\Omega(\rej/\clH_\ell) = \Prob_{\clH_\ell}\{(\underline x, \underline y) \in \cup_{j \in \{1, \ldots, M\} \setminus \{\ell\}} \Omega_j \}
%P_\Omega(\rej/\clH_\ell) = \Prob_{\clH_\ell}\left\{(\underline x, \underline y) \in \Omega_R \right\}
%\end{equation}
%for problem~\Prb{unknown}.
The rejection exponents capture the rate of decay of the rejection probabilities and are defined as
\begin{equation}
\liminf_{n \to \infty} -\frac{1}{n} \log P_\Omega(\rej/\clH_\ell).\label{eqn:rejexpdefn}
\end{equation}

Some special cases of problem~\Prb{known} are listed below.
\begin{enumerate}
%\item If $M=2$ and $N=K=1$, this is just the classical binary hypothesis testing problem.
\item If $M\geq2$ and $N=K=1$, this is just the classical $M$-ary hypothesis testing problem, with rejection.
\item For general $M,N,K$, a variant of this hypothesis testing problem without rejection is discussed in \cite{ahlhar06}.
The special case of $M=N=K=2$ is considered in detail.
In this case there are exactly two hypotheses.
The authors solve the problem of optimizing one type of error exponent under a constraint on the other type of error exponent.
\end{enumerate}

Specific versions of problem~\Prb{unknown} have been studied in the past.
Some special cases of interest are listed below.
\begin{enumerate}
\item When $N=K=1$, this is the problem studied by Gutman in \cite{gut89}.
The results and approach of the present paper are largely based on \cite{gut89}.
\item For $M=N=K$ we studied problem~\Prb{unknown} in a recent work \cite{unnnai13}.
\end{enumerate}
In the present paper, we generalize these works to arbitrary choices of $\clM$, $\clN$ and $\clK$.

The main tools we use for proving the results in this paper are the \emph{method of types} \cite{covtho06} and Sanov's theorem \cite{san57a} (see also \cite{demzei98a}).
The following lemma (see e.g., \cite[Ch. 11]{covtho06} for a proof) gives a bound on the probability of observing a sequence with a specific type, or equivalently, a specific empirical distribution.
\begin{lemma}\label{lem:typeprobnew}
Let $\ystate$ be a finite set and $s \in \ystate^n$ be an arbitrary string of length $n$ with entries in $\ystate$.
Let $y \in \ystate^n$ be a random string drawn i.i.d. under probability law $\nu \in \py$.
Then
\begin{equation}
%2^{-n \left(D(\Gamma_s \| \nu) + \frac{|\zstate| \log (n+1)}{n}\right)} = \frac{1}{(n+1)^{|\zstate|}} 2^{-n D(\Gamma_s \| \nu)} \leq \Prob\{\Gamma_y = \Gamma_s\} \leq 2^{-n D(\Gamma_s \| \nu)} \label{eqn:typeineq}
2^{-n \left(D(\Gamma_s \| \nu) + \frac{|\zstate| \log (n+1)}{n}\right)} \leq \Prob\{\Gamma_y = \Gamma_s\} \leq 2^{-n D(\Gamma_s \| \nu)}\label{eqn:typeineq}
\end{equation}
where $\Gamma_s$ and $\Gamma_y$ represent the empirical distributions of $s$ and $y$ respectively.
\qed
%where $\Prob$ denotes the probability measure when all observations in $y$ are drawn i.i.d. according to law $p$.
\end{lemma}
Sanov's theorem is a statement on the behavior of the probability as $n \to \infty$.
It characterizes the large deviations behavior of the empirical distribution of an i.i.d. sequence as stated below.
\begin{theorem}[Sanov \cite{san57a}]\label{thm:san}
Let $\ystate$ be a finite set.
For any $\nu \in \py$ if $y \in \ystate^n$ is a random sequence of length $n$ drawn i.i.d. under $\nu$, and $A \subset \py$
%\begin{equation}
%%(n+1)^{|\zstate|} 2^{-n D(\mu \| \nu)} \leq \Prob \{\Gamma_y = \mu \} \leq 2^{-n D(\mu \| \nu)}.
%%\Prob \{\Gamma_y \in A\} \leq (n+1)^{|\zstate|} 2^{-n D(\mu^* \| \nu)} =  2^{-n \left(D(\mu^* \| \nu)    - \frac{|\zstate| \log (n+1)}{n}\right)}
%\Prob \{\Gamma_y \in A\} \leq (n+1)^{|\zstate|} 2^{-n D(\mu^* \| \nu)} %=  2^{-n \left(D(\mu^* \| \nu)    - \frac{|\zstate| \log (n+1)}{n}\right)}
%\end{equation}
%where $\mu^* = \argmin_{\mu \in A} D(\mu \| \nu)$.
%Furthermore, if $A$ is the closure of its interior, then
%\[
%\lim_{n \to \infty} \log \Prob \{\Gamma_y \in A\} = - D(\mu^* \| \nu).
%\]
such that $A$ is the closure of its interior, then
\[
\lim_{n \to \infty} \frac{1}{n}\log \Prob \{\Gamma_y \in A\} = - \min_{\mu \in A} D(\mu \| \nu).%D(\mu^* \| \nu).
\]
%where $\mu^* = \argmin_{\mu \in A} D(\mu \| \nu)$.
\qed
\end{theorem}
The main result of Section~\ref{sec:optunknown} is based on the following lemma which gives a bound on large-deviations of a pair of empirical distributions.
\begin{lemma}\label{lem:ldptwoemp}
Let $\ystate$ be a finite set.
For $i =1,2$ let $y_i \in \ystate^n$ denote a length $n$ string drawn i.i.d. under $\nu \in \py$.
Further assume that $y_1$ and $y_2$ are mutually independent.
Then we have
\[
%\lim_{n \to \infty} - \frac{1}{n} \log \Prob\{D(\Gamma_{y_1} \| \half(\Gamma_{y_1}  + \Gamma_{y_2} ))  + D(\Gamma_{y_2} \| \half(\Gamma_{y_1}  + \Gamma_{y_2} )) \geq \lambda\} \geq \lambda.
\lim_{n \to \infty} - \frac{1}{n} \log \Prob\{\sum_{i=1}^2D(\Gamma_{y_i} \| \half(\Gamma_{y_1}  + \Gamma_{y_2} ))   \geq \lambda\} \geq \lambda.
\]
\qed
\end{lemma}
We provide a proof in the appendix.

It is possible to generalize Sanov's theorem to the infinite alphabet setting in which $\ystate$ is countably or uncountably infinite (see, e.g., \cite{demzei98a}).
However, in this paper we focus only on the finite alphabet setting.
The analysis of error probabilities and optimal tests in the continuous alphabet setting is much more involved, and are typically based on Cramer's theorem \cite{demzei98a}.
We present discussions on potential extensions of the results of this paper to other settings, including that of continuous alphabets, in Section~\ref{sec:conc}.

%the empirical distributions $\Gamma_y$ appearing in Theorem~\ref{thm:san} are always supported on a finite set and hence

Before we present the solutions, we summarize the main results below.
\begin{itemize}
\item Optimal test for the matching problem~\Prb{known} with known source distributions, given in Theorem~\ref{thm:opt1}.
\item Comparison of error exponents and rejection exponents of the optimal test with the test that ignores constraints on sequences.
%\item Comparison of rejection probabilities of optimal test with the test that ignores constraints on sequences, given in Proposition~\ref{prop:rejprobcompareknown}.
%\item Comparison of error probabilities of the optimal test without rejection, with the test without rejection that ignores constraints on sequences, given in Proposition~\ref{prop:errexpnorej}.
\item Optimal test for the matching problem~\Prb{unknown} with unknown source distributions, given in Theorem~\ref{thm:opt}.
\end{itemize}

%%%%%%%%%%%%%%%%%%%%%%%%%%%%%%%%%%%%%%%%%%%%%%%%%%%%%%%%%%%%%%%%%
%%%%%%%%%%%%%%%%%%%%%%%%%%%%%%%%%%%%%%%%%%%%%%%%%%%%%%%%%%%%%%%%%
%%%%%%%%%%%%%%%%%%%%%%%%%%%%%%%%%%%%%%%%%%%%%%%%%%%%%%%%%%%%%%%%%

\section{Optimal Matching with known distributions}\label{sec:optknown}
In this section we solve problem~\Prb{known}.
As described in Section~\ref{sec:problemstatement} we use the optimality criterion based on error exponents given in \eqref{eqn:criterion}.
Let $\clM = \{\mu_1,\mu_2,\ldots,\mu_M\} \subset \pz$ as before.
Let $\clG$ be the graph in Figure~\ref{fig:compbip} with $\clV_1$ and $\clV_2$ respectively representing $\clM$ and $\clS$, both of which are described in the statement of problem~\Prb{known}.
Let $\match_\ell \subset \{1,2,\ldots,M\} \times \{1,2,\ldots,N\}$ with $|\match_\ell| = K$ denote the matching on the complete bipartite graph $\clG$ under hypothesis $\clH_\ell$.
Any edge $e \in \match_\ell$ can be represented as $e = (e_1,e_2)$ with the understanding that the edge connects $\mu_{e_1}$ and $y_{e_2}$ in graph $\clG$.
%For each $e \in \match_\ell$ let $\mu_e$ denote the distribution of the source that produced sequences $x_{e_1}$ and $y_{e_2}$.
Thus, under hypothesis $\clH_\ell$ we have
\[
\clK = \{\mu_{e_1}: e \in \match_\ell\} = \clM \cap \clN
\]
representing the probability distributions followed by the sources that produced the sequences in $\{y_{e_2}: e \in \match_\ell\}$.
There are $M-K$ sources in $\clM \setminus \clK$, which do not produce any sequence in $\clS$, and there are $N-K$ sequences in $\clS$ that are produced by sources in $\clN \setminus \clK$.

%There are $M-K$ sources in $\clM \setminus \clK$, whose indices are collectively represented as
%\[
%I_\mu^\ell = \{j:\mbox{ No edge in }\match_\ell\mbox{ is incident on }\mu_j\}.
%\]
%Similarly, there are $N-K$ sequences in $\clS$ that are produced by sources in $\clN \setminus \clK$.
%The indices of these sequences are represented via the following notation:
%\[
%I_y^\ell = \{j:\mbox{ No edge in }\match_\ell\mbox{ is incident on }y_j\}.
%\]
%Furthermore, we use $q^y_i$ to denote the probability distribution of the source that produced sequence $y_i$.
%Thus under hypothesis $\clH_\ell$, we have
%\[
%\clM = \{\mu_j: j \in I_\mu^\ell\} \cup \clK
%\]
%and
%\[
%\clN = \{q^y_j: j \in I_y^\ell\} \cup \clK.
%\]

Let
\begin{eqnarray}
D(\clH_\ell)&=& \sum_{(i,j) \in \match_\ell} D(\Gamma_{y_j} \| \mu_i).\label{eqn:matchweight}
\end{eqnarray}
Consider the estimate for the hypothesis given by
\begin{equation}
\what \clH = \what \clH(\underline{y}) = \argmin_{\clH_\ell} D(\clH_\ell)\label{eqn:optmatch1}
\end{equation}
where the minimization is performed over all hypotheses.
As each hypothesis is represented by a matching on $\clG$ with cardinality equal to $K$, the estimate of \eqref{eqn:optmatch1} can be interpreted as the hypothesis corresponding to the minimum weight cardinality-$K$ matching \cite{west2001introduction} on $\clG$ with appropriate weights assigned to the edges in $\clG$.
For $\mu_i \in \clM$ and $y_j \in \clS$ we let the weight $w_{ij}$ of the edge between them to be %****notation for edges is not consistent!!!*****
\begin{equation}
w_{ij} = D(\Gamma_{y_j} \| \mu_i).\label{eqn:weights1} %D(\Gamma_{x_i} \| \half({\Gamma_{x_i} + \Gamma_{y_j}}) )+ D(\Gamma_{y_j} \| \half({\Gamma_{x_i} + \Gamma_{y_j}})).\label{eqn:weights1}
\end{equation}
Weight $w_{ij}$ can be interpreted as a measure of the difference between distributions $\mu_i$ and $\Gamma_{y_j}$.
Figure~\ref{fig:weigtedbip1} shows the graph $\clG$ with weights added to the edges.
%\begin{figure}[h]
%\centering
%%\psfrag{a}[Bc][Bc]{\small{$x_1$}}
%%\psfrag{b}[Bc][Bc]{\small{$x_2$}}
%%\psfrag{c}[Bc][Bc]{\small{$y_1$}}
%%\psfrag{d}[Bc][Bc]{\small{$y_2$}}
%%\psfrag{e}[Bc][Bc]{\small{$x_K$}}
%%\psfrag{f}[Bc][Bc]{\small{$y_K$}}
%%\psfrag{w}[Bc][Bc]{\small{$w_{11}$}}
%%\psfrag{\clS}[Bc][Bc]{\small{Set $\clS_1$}}
%%\psfrag{Y}[Bc][Bc]{\small{Set $\clS_2$}}
%\includegraphics[width=.6\columnwidth]{FIG/bipartitematch.eps}
%\caption{The solution to the mutiple hypothesis testing problem given in~\eqref{eqn:optmatch1} can be obtained by performing a minimum weight bipartite matching with weights given in~\eqref{eqn:weights1}.} \label{fig: bipartitematch}
%\end{figure}

\begin{figure}[h]
\centering
\psfrag{s2}{$\clV_1 = \clM$}
\psfrag{o2}{$\clV_2 = \clS$}

\psfrag{p1}{$\mu_1$}
\psfrag{p2}{$\mu_2$}
\psfrag{p3}{$\mu_i$}
\psfrag{p4}{$\mu_M$}
\psfrag{vd}{\Huge $\vdots$}
\psfrag{y1}{$y_1$}
\psfrag{y2}{$y_2$}
\psfrag{y3}{$y_j$}
\psfrag{y4}{$y_N$}
\psfrag{wij}{$w_{ij} = D(\Gamma_{y_j} \| \mu_i)$}
\includegraphics[width=.4\columnwidth]{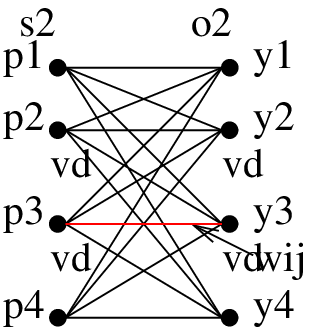}
\caption{The weighted complete bipartite graph $\clG$ for problem~\Prb{known}. The weight of the edge between the $i$-th vertex in $\clV_1$ and the $j$-th vertex in $\clV_2$ is given by~\eqref{eqn:weights1}. The matching corresponding to the hypothesis $\what \clH$ in~\eqref{eqn:optmatch1} is given by the minimum weight matching on this graph with cardinality $K$.} \label{fig:weigtedbip1}
\end{figure}

%The solution of \eqref{eqn:optmatch1} can be justified by the asymptotic optimality properties of a threshold test that uses this statistic.
%
We will now show that a test based on the estimate of $\what \clH$ in \eqref{eqn:optmatch1} is asymptotically optimal.
For proving optimality we restrict ourselves to tests that are based only on the empirical distributions of the observations.
Let $\Gamma_Y$ denote the collection of empirical distributions:
\[
\Gamma_Y:= \left(\Gamma_{y_1},\Gamma_{y_2},\ldots,\Gamma_{y_N}\right).
\]
The restriction to tests based on empirical distributions is justified in the asymptotic setting because of the following lemma.
\begin{lemma}\label{lem:typesuff1}
Let $\Omega = (\Omega_1, \Omega_2,  \ldots, \Omega_J, \Omega_R)$ be a decision rule based only on the distributions $\{\mu_1,\mu_2,\ldots,\mu_M\}$ and $\{y_1,y_2,\ldots,y_N\}$. Then there exists a decision rule $\Lambda = (\Lambda_1, \Lambda_2, \ldots, \Lambda_J, \Lambda_R)$ based on the sufficient statistics $\Gamma_Y$ such that
%\begin{eqnarray}
%&&\limsup_{n \to \infty} \frac{1}{n} \log P_\Lambda(\err/\clH_\ell) \leq  \limsup_{n \to \infty} \frac{1}{n} \log P_\Omega(\err/\clH_\ell), \nonumber\\
%&&\quad \ell=1,2,\ldots,J, \forall \clU \subset \pz \label{eqn:lemclaim1}
%\end{eqnarray}
\begin{eqnarray*}
\liminf_{n \to \infty} -\frac{1}{n} \log P_\Lambda(\err/\clH_\ell) &\geq&  \liminf_{n \to \infty} -\frac{1}{n} \log P_\Omega(\err/\clH_\ell), \\%\label{eqn:lemclaim1}\\
%\quad \ell \in [J], \forall \clU \subset \pz, \\
%\end{eqnarray}
%and
%\begin{eqnarray}
%\limsup_{n \to \infty} \frac{1}{n} \log |\Lambda_R| \leq \limsup_{n \to \infty} \frac{1}{n} \log |\Omega_R|.\label{eqn:lemrejsize1}
\liminf_{n \to \infty} -\frac{1}{n} \log P_\Lambda(\rej/\clH_\ell) &\geq&  \liminf_{n \to \infty} -\frac{1}{n} \log P_\Omega(\rej/\clH_\ell),% \label{eqn:lemclaimrejprob}
%\quad \ell \in [J], \forall \clU \subset \pz. 
\end{eqnarray*}
for all $\ell \in [J]$ and for all choices of $\clU$.
\qed
\end{lemma}
%RRRR
We provide a proof in the appendix.
Thus this lemma suggests that if one is interested only in optimizing error exponents and rejection exponents, then tests based only on $\Gamma_Y$ are sufficient.
%We will henceforth focus on tests that are based only on $\Gamma_Y$.
%Note that $\Lambda_R$ and $\Omega_R$ are finite sets, thus their cardinality is well-defined.

Following Gutman \cite{gut89}, in order to prove optimality we allow for a no-match zone, i.e., we allow a decision in favor of rejecting all the $M$ hypotheses.
For this purpose, we need to identify the hypothesis corresponding to the second minimum weight matching in $\clG$.
Let
\begin{equation}
\wtilde \clH  = \wtilde \clH(\underline{x},\underline{y}) = \argmin_{\clH_\ell \neq \what \clH} D(\clH_\ell)\label{eqn:optmatch31}
\end{equation}
where $\what \clH$ is defined in~\eqref{eqn:optmatch1}.
The choices of $\what \clH$ and $\wtilde \clH$ have a simple interpretation in terms of maximum \emph{generalized likelihoods} \cite{leh05} as shown in the lemma below.
\begin{lemma}\label{lem:ML}
The selections $\what \clH$ defined in \eqref{eqn:optmatch1} and $\wtilde \clH$ defined in \eqref{eqn:optmatch31} can be expressed as
\begin{eqnarray}
\what \clH = \clH_{\what \ell}  \quad \mbox{ and } \quad \wtilde \clH = \clH_{\wtilde \ell} \label{eqn:HhatML}
\end{eqnarray}
where
\begin{eqnarray*} 
\what \ell &=& \argmax_{\ell \in [J]} \max_{\substack{\clN \subset \pz\\\clM \cap \clN = \clK} }\Prob_{\clH_\ell} (y_1, y_2, \ldots, y_N)\\
\wtilde \ell &=&  \argmax_{\ell \in [J]: \clH_\ell \neq \what \clH} \max_{\substack{\clN \subset \pz\\\clM \cap \clN = \clK} }\Prob_{\clH_\ell} (y_1, y_2, \ldots, y_N).
\end{eqnarray*}
%
%\begin{eqnarray}
%%\what \clH &=& \clH_{\what \ell}  \mbox{ where }  \what \ell = \argmax_{\ell \in [J]} \max_{\substack{\clN \subset \pz\\\clM \cap \clN = \clK} }\Prob_{\clH_\ell} (y_1, y_2, \ldots, y_N), \label{eqn:HhatML}\\
%%\wtilde \clH &=& \clH_{\wtilde \ell} \quad \mbox{ where } \quad \wtilde \ell =  \argmax_{\ell \in [J]: \clH_\ell \neq \what \clH} \max_{\substack{\clN \subset \pz\\\clM \cap \clN = \clK} }\Prob_{\clH_\ell} (y_1, y_2, \ldots, y_N). \label{eqn:HtildeML}
%\what \clH &=& \clH_{\what \ell}  \mbox{ where }  \what \ell = \argmax_{\ell \in [J]} \max_{\substack{\clN \subset \pz\\\clM \cap \clN = \clK} }\Prob_{\clH_\ell} (y_1, y_2, \ldots, y_N), \label{eqn:HhatML}\\
%\wtilde \clH &=& \clH_{\wtilde \ell} \quad \mbox{ where } \quad \wtilde \ell =  \argmax_{\ell \in [J]: \clH_\ell \neq \what \clH} \max_{\substack{\clN \subset \pz\\\clM \cap \clN = \clK} }\Prob_{\clH_\ell} (y_1, y_2, \ldots, y_N). \label{eqn:HtildeML}
%\end{eqnarray}
\qed
\end{lemma}
The above lemma is proved in the appendix.
It is easy to be see that in the special case that $M \geq N = K$, the set $\clN$ is fixed and thus the second minimization over the choice of distributions in $\clN$ is not necessary.
In such a case the choice of $\what \clH$ can be interpreted as a simple maximum likelihood hypothesis.
The optimal test with rejection can be described in terms of $\what H$ and $\wtilde H$ as shown in the following theorem.
%\subsection{Optimality properties}
%We now proceed to characterize the optimality properties enjoyed by the solution of \eqref{eqn:optmatch}.
\begin{theorem}\label{thm:opt1}
Let $\clM = \{\mu_1,\mu_2,\ldots,\mu_M\} \subset \pz$ be a known set of $M$ distinct probability distributions on the finite alphabet $\zstate$.
Let $\Omega = (\Omega_1, \Omega_2,  \ldots, \Omega_J, \Omega_R)$ be a decision rule based on the collection $\Gamma_Y$ of empirical distributions
%sequences $\{x_1,\ldots,x_K\}$ and $\{y_1,\ldots,y_K\}$,
such that %for all choices of distributions in $\clM = \{\mu_1,\mu_2,\ldots,\mu_M\} \subset \pz$, we have
\begin{eqnarray}
P_\Omega(\err/\clH_\ell) \leq 2^{-\lambda n}, \mbox{ for all }\ell \in [J] \label{eqn:errprobconstraint}
\end{eqnarray}
and for all choices of distributions in $\clN \setminus \clK$.
%when source $k$ is distributed according to $p_k$ for $k \in \{1,2,\ldots,K\}$.

Let $\wtilde \lambda = \lambda - \frac{N|\zstate| \log(n+1)}{n}$,
\begin{eqnarray*}
\Lambda_\ell =  \{\underline y: D(\wtilde \clH) \geq \wtilde \lambda, \what \clH = \clH_\ell\}, \ell \in [J],
\end{eqnarray*}
and
\[
\Lambda_R = \{\underline y: D(\wtilde \clH) < \wtilde \lambda\}.
\]
Then
%\begin{eqnarray}
%&&\limsup_{n \to \infty} \frac{1}{n} \log P_\Lambda(\err/\clH_\ell) \leq - \lambda, \nonumber\\
%&&\ell=1,2,\ldots,J, \forall \clU \subset \pz \label{eqn:thmclaim11}
%\end{eqnarray}
\begin{eqnarray}
\liminf_{n \to \infty} -\frac{1}{n} \log P_\Lambda(\err/\clH_\ell) \geq \lambda, \quad \ell \in [J], \forall \clU \subset \pz \label{eqn:thmclaim11}
\end{eqnarray}
and
\begin{eqnarray}
\Lambda_R \subset \Omega_R. \label{eqn:thmclaim21}
\end{eqnarray}
\qed
\end{theorem}
We provide a proof to the theorem in the appendix.
From the definition of $D(\clH_\ell)$ in \eqref{eqn:matchweight} it is clear that the decision regions $\Lambda_\ell$'s and $\Lambda_R$ proposed in Theorem \ref{thm:opt1} depends on the sequences in $\underline y$ only through $\Gamma_Y$.
An illustration of the various decision regions of the optimal test as functions of $\Gamma_Y$ is provided in Figure~\ref{fig:regions} for a specific example.
From the condition of \eqref{eqn:thmclaim21} it is obvious that the probability of rejection of the decision rules under the decision rule $\Lambda$ is lower than the probability of rejection under the decision rule $\Omega$.
Thus the optimality result implies that the test $\Lambda$ has lower rejection probabilities, as defined in \eqref{eqn:rejprob}, than any test $\Omega$ that satisfies an exponential decay of error probabilities as in \eqref{eqn:errprobconstraint}.

%\begin{figure}[h]
%\centering
%\psfrag{p1}{}%{$\pi^{\sigma^1}$}
%\psfrag{p2}{}%{$\pi^{\sigma^2}$}
%\psfrag{p3}{}%{$\pi^{\sigma^3}$}
%\psfrag{p4}{}%{$\pi^{\sigma^4}$}
%\psfrag{Q1}{$S_1$}
%\psfrag{Q2}{$S_2$}
%\psfrag{Q3}{$S_3$}
%\psfrag{Q4}{$S_4$}
%\psfrag{L1}{$\{\Gamma_Y: \underline y \in \Lambda_1\}$}
%\psfrag{L2}{$\{\Gamma_Y: \underline y \in \Lambda_2\}$}
%\psfrag{L3}{$\{\Gamma_Y: \underline y \in \Lambda_3\}$}
%\psfrag{L4}{$\{\Gamma_Y: \underline y \in \Lambda_4\}$}
%\psfrag{LR}{$\{\Gamma_Y: \underline y \in \Lambda_R\}$}
%\psfrag{B}{$\{\Gamma_Y: D(\clH_1) = D(\clH_2)\}$}
%%\psfrag{pz}{$(\zstate^n)^N$}
%\psfrag{pz}{$(\pz)^N$}
%\includegraphics[width=.7\columnwidth]{Figures/Regions3.eps}
%\caption{Illustration of the decision regions for the optimal test of Theorem~\ref{thm:opt1}. Here $S_\ell =\{\Gamma_Y: D(\clH_\ell) < \lambda\}$.} \label{fig:regions}
%\end{figure}

\begin{figure}[h]
\centering
%\psfrag{p1}{}%{$\pi^{\sigma^1}$}
%\psfrag{p2}{}%{$\pi^{\sigma^2}$}
%\psfrag{p3}{}%{$\pi^{\sigma^3}$}
%\psfrag{p4}{}%{$\pi^{\sigma^4}$}
\psfrag{p1}{$(\mu_1,\mu_2)$}%{$\pi^{\sigma^1}$}
\psfrag{p2}{$(\mu_2,\mu_1)$}%{$\pi^{\sigma^2}$}
\psfrag{Q1}{$T_1$}
\psfrag{Q2}{$T_2$}
\psfrag{Q3}{$T_3$}
\psfrag{Q4}{$T_4$}
\psfrag{L1}{$L_1$}%{$\{(\nu_1,\nu_2,\ldots, \nu_N): \sum_{(i,j) \in \match_1} D(\nu_{j} \| \mu_i) < \lambda\}$}
\psfrag{L2}{$L_2$}%{$\{(\nu_1,\nu_2,\ldots, \nu_N): \sum_{(i,j) \in \match_2} D(\nu_{j} \| \mu_i)\}$}
\psfrag{L3}{$L_3$}%{$\{(\nu_1,\nu_2,\ldots, \nu_N): \sum_{(i,j) \in \match_3} D(\nu_{j} \| \mu_i)\}$}
\psfrag{L4}{$L_4$}%{$\{(\nu_1,\nu_2,\ldots, \nu_N): \sum_{(i,j) \in \match_4} D(\nu_{j} \| \mu_i)\}$}
\psfrag{LR}{$L_R$}%{$\{(\nu_1,\nu_2,\ldots, \nu_N): \}$}
\psfrag{B}{$B_{12}$}
%\psfrag{B}{$\{(\nu_1,\nu_2,\ldots, \nu_N): \sum_{(i,j) \in \match_1} D(\nu_{j} \| \mu_i) = \sum_{(i,j) \in \match_1} D(\nu_{j} \| \mu_i)\}$}
%\psfrag{pz}{$(\zstate^n)^N$}
\psfrag{pz}{$(\pz)^2$}
\includegraphics[width=.6\columnwidth]{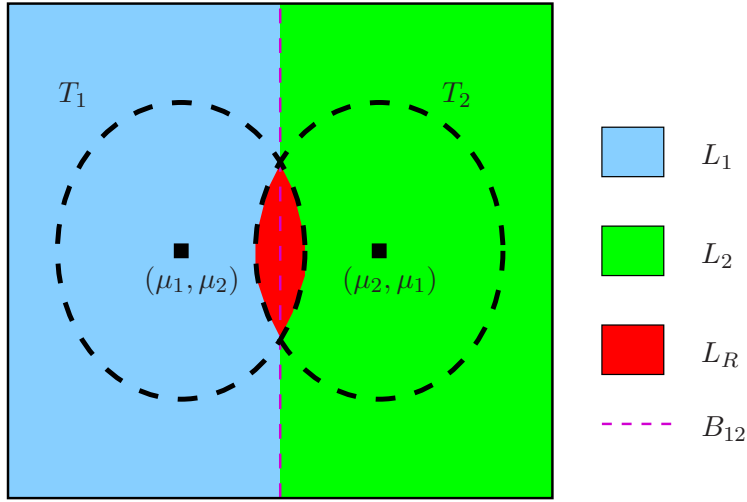}
\caption{Illustration of the decision regions for the optimal test of Theorem~\ref{thm:opt1} for $M=N=K=2$, which means that $J=2$. Here we define $T_\ell =\{(\nu_1,\nu_2): \sum_{(i,j) \in \match_\ell} D(\nu_{j} \| \mu_i)  < \lambda\}$.
Assuming $\match_1$ denotes the matching in which $y_i$ is matched to $\mu_i$, and $\match_2$ denotes the matching in which $y_1$ is matched to $\mu_2$ and $y_2$ to $\mu_1$, it follows that $(\mu_1,\mu_2) \in T_1$ and $(\mu_2,\mu_1) \in T_2$ .
We use $B_{12}$ to denote the hyperplane separating $L_1$ and $L_2$ defined as $B_{12}= \{(\nu_1,\nu_2): \sum_{(i,j) \in \match_1} D(\nu_{j} \| \mu_i) = \sum_{(i,j) \in \match_2} D(\nu_{j} \| \mu_i)\}$, or equivalently, $B_{12}= \{(\nu_1,\nu_2): \sum_{z \in \zstate} (\nu_1(z) - \nu_2(z)) \log \frac{\mu_1(z)}{\mu_2(z)} = 0\}$. %The separating hyperplanes between other $L_i$ are defined analogously.
Then the optimal decision regions of Theorem~\ref{thm:opt1} can be expressed as $\Lambda_i = \{\underline y: \Gamma_Y \in L_i\}$ for $i \in \{1,2,R\}$, where $L_i$ are as shown in the figure.
Furthermore, if $\Gamma_Y$ lies to the left of $B_{12}$ in the figure, then $\what \clH = \clH_1$ and if $\Gamma_Y$ lies to the right of $B_{12}$ then $\what \clH = \clH_2$.}\label{fig:regions}
\end{figure}

The test can be explained in words as follows.
First identify the hypotheses corresponding to the minimum weight matching and the second minimum weight matching of cardinality $K$ in $\clG$.
Accept the former hypothesis if the weight corresponding to the latter exceeds the threshold $\wtilde \lambda$, and reject all hypotheses if the threshold is not exceeded.
When $M \geq N =K$, the result of Lemma~\ref{lem:ML} implies that this test leads to a rejection if the weights corresponding to the two most likely hypotheses, $\what \clH$ and $\wtilde \clH$, are below a threshold, or equivalently, if the observations can be  well-explained by two or more hypotheses.

We note that the threshold $\wtilde \lambda$ appearing in the definition of $\Lambda_R$  satisfies $\wtilde \lambda \to \lambda$ as $n \to \infty$.
Using Sanov's theorem we show in the proof that the choice of decision regions ensures that the error-exponent constraint of \eqref{eqn:thmclaim11} is satisfied.
We also observe that the offset between $\wtilde \lambda$ and $\lambda$ is just $N$ times the offset in the exponent appearing in the first inequality of \eqref{eqn:typeineq} which bounds the probability of observing a type.
This offset is introduced to ensure that the condition \eqref{eqn:thmclaim21} is satisfied, as detailed in the proof.

\subsection{Comparison with the unconstrained problem}\label{sec:compknown}
As we mentioned earlier, the problem studied in this section differs from ordinary multiple hypothesis testing because of the prior knowledge that the strings in $\clS$ were generated by distinct sources.
It is interesting to compare the performance obtained by using the optimal test that makes use of this information with the performance of the optimal test that one would have to use in the absence of this prior information.
Before we proceed we need to introduce some new notations.
Let $\pi, \pi_1,\pi_2 \in \clP(\ystate)$ be distinct probability mass functions with complete supports on $\ystate$.
For any $\eta \geq 0$ we define
\begin{equation}
Q_\eta(\pi) := \{\nu \in \clP(\ystate): D(\nu \| \pi) < \eta \}, \label{eqn:qdefn}
\end{equation}
and
\begin{equation}
E_\eta({\pi_1,\pi_2}) = \sup \{\beta \geq 0: Q_\beta(\pi_2)\cap Q_\eta(\pi_1) =\emptyset\}. \label{eqn:errexp}
\end{equation}
The function $E_\eta({\pi_1,\pi_2})$ is strictly monotonically decreasing in $\eta$ in the interval $\eta \in (0, D(\pi_2\|\pi_1))$ (see, e.g., \cite[Sec. 11.7]{covtho06}, \cite[Sec. 3.2]{lev08}, and \cite{bla74}).
Moreover, if $\eta \geq D(\pi_2 \| \pi_1)$, then $E_\eta({\pi_1,\pi_2}) = 0$.
The Chernoff information \cite{covtho06} between $\pi_1$ and $\pi_2$ is defined as
\[
C(\pi_1,\pi_2):= - \inf_{\alpha \in [0,1]} \log\left( \sum_{y \in \ystate} \pi_1^\alpha(y) \pi_2^{1-\alpha}(y) \right).
\]
It is well known \cite{covtho06,lev08} that
\[
E_{C(\pi_1,\pi_2)}({\pi_1,\pi_2}) = C(\pi_1,\pi_2).%, \mbox{ for all } \pi_1,\pi_2 \in \clP(\ystate).
\]
%Moreover% if $\pi_1$ and $\pi_2$ have complete support over $\ystate$, then
%\begin{equation}
%Q_\eta(\pi_1) \cap Q_\eta(\pi_2) = \emptyset \mbox{ for all } \eta \leq C(\pi_1,\pi_2)\label{eqn:nonintersect}
%\end{equation}
%and
%\begin{equation}
%Q_\eta(\pi_1) \cap Q_\eta(\pi_2) \neq \emptyset \mbox{ for all } \eta > C(\pi_1,\pi_2)\label{eqn:intersect}
%\end{equation}
These quantities are illustrated in Figure~\ref{fig:qballs}.
%As the Kullback-Leibler divergence is a convex function of its arguments, it follows that $Q_\eta(\pi)$ is a convex set.

\begin{figure}[h]
\centering
\psfrag{p1}{$\pi_1$}
\psfrag{p2}{$\pi_2$}
\psfrag{Q1}{$Q_\eta(\pi_1)$}
\psfrag{Q2}{$Q_\beta(\pi_2)$}
\psfrag{Pz}{$\py$}
\includegraphics[width=.4\columnwidth]{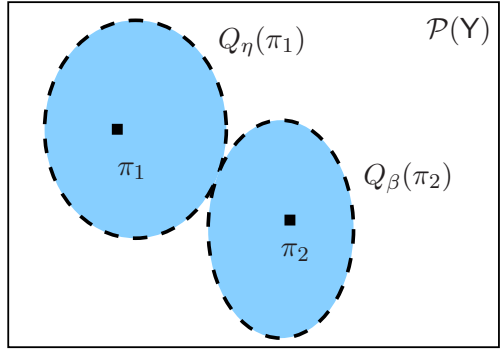}
\caption{Illustration of the $Q_\eta$ sets defined in \eqref{eqn:qdefn} on the probability simplex $\py$. In the above figure, the sets $Q_\eta(\pi_1)$ and $Q_\beta(\pi_2)$ touch each other and hence $\beta = E_\eta(\pi_1,\pi_2)$ as defined in \eqref{eqn:errexp}. If in addition $\beta = \eta$ then $\eta = C(\pi_1,\pi_2)$.} \label{fig:qballs}
\end{figure}

Below we provide two comparisons, the first being a comparison of the rejection regions in the two cases, and the second, a comparison of the error exponents obtained by tests that do not allow for a rejection region.

\subsubsection{Comparison of rejection regions}\label{sec:compknownrej}
%For the purpose of comparison we restrict ourselves to the setting in which $M \geq N = K$, although the ideas presented here can be extended to more general settings.
In the absence of prior knowledge that the strings in $\clS$ were generated by distinct sources, one is forced to repeatedly perform optimal multihypothesis testing on each string in $\clS$.
In other words, for each string $s$ in $\clS$, one repeats the optimal solution of Theorem \ref{thm:opt1} assuming that the second set $\clS$ is a singleton comprising the single string $s$.
These individual solutions can then be combined to obtain a solution to the original problem as follows.

Assume $M \geq N = K$.
Consider the function $\widehat \sigma: [N] \mapsto [M]$ defined by
%For each $i \in [N]$ define
\begin{equation}
\widehat \sigma(i) = \argmin_{j\in [M]} D(\Gamma_{y_i}\|\mu_j), \mbox{ for all } i \in [N]. \label{eqn:ucsigmahat}
\end{equation}
Thus the function $\widehat \sigma$ gives the best matching of each string to one of the sources in $\clM$.
In fact, it can be shown that $\widehat \sigma(i)$ is the maximum likelihood source that produced $y_i$, just as in Lemma~\ref{lem:ML}.
If $\what \sigma$ is a one-to-one function, then it corresponds to a valid hypothesis for the matching problem.
We call this hypothesis $\clH^{\widehat \sigma}$.
If ${\widehat \sigma}$ is not a one-to-one function, or if $\clH^{\widehat \sigma}$ does not correspond to the true hypothesis, then the strings are not correctly matched and hence in this case an error occurs.
Furthermore, in order to satisfy the error exponent constraint, one is forced to reject whenever the individual hypothesis test on any of the $N$ strings leads to a rejection.
%For each $i$ we identify the weight of the second lightest edge in $\clG$ incident at $s_i$.
Let
\[
\widetilde w_i = \min_{j\in [M]\setminus{{\widehat \sigma}(i)}} D(\Gamma_{y_i}\|\mu_j), \quad i \in [N]
\]
The solution to the overall problem is now given by
\begin{eqnarray}
\Lambda_\ell^{\uc} =  \{\underline y: \min_{i \in [N]} \widetilde w_i \geq \widecheck \lambda, \clH^{\widehat \sigma}= \clH_\ell\}, \ell \in [J]\label{eqn:decreguc}
\end{eqnarray}
where the superscript of $\uc$ indicates that the solution is unconstrained and
\begin{eqnarray}
\Lambda_R^{\uc} = \{\underline y: \min_{i \in [N]} \widetilde w_i < \widecheck \lambda\} \label{eqn:rejuc}% \cup \{\Gamma_Y: {\widehat \sigma} \mbox{ is not a one-one-function}\}
\end{eqnarray}
where $\widecheck \lambda = \lambda - \frac{|\zstate| \log(n+1)}{n}$, is the optimal choice for the threshold obtained from Theorem \ref{thm:opt1} when the set of strings is a singleton.
%It is to be noted that in this case $\bigcup_{\ell \in [J]} \Lambda_\ell^{\uc} \bigcup \Lambda_R^{\uc} \neq \mathbf{Z}^1$ because it is possible that ${\widehat \sigma}$ might not be a permutation.
The probability of error of this solution is given by
\begin{eqnarray*}
\lefteqn{P_{\Lambda^{\uc}}(\err|\clH_\ell)} \nonumber\\
 &=& \Prob_{\clH_\ell}\{y_i \mbox{ is incorrectly matched for some }i \in [N]\}\\
&\leq& \sum_{i=1}^N \Prob_{\clH_\ell}\{y_i \mbox{ is incorrectly matched}\}
\end{eqnarray*}
where the inequality follows via the union bound.
By the result of Theorem \ref{thm:opt1}, each term in the above summation decays exponentially in $n$ with exponent $\lambda$ and thus we have
\[
\liminf_{n \to \infty} -\frac{1}{n} \log P_{\Lambda^{\uc}}(\err|\clH_\ell) \geq \lambda.
\]
Thus this solution meets the same error exponent constraint as the optimal solution of Theorem \ref{thm:opt1} that one can use when the constraint on the strings is known a priori.
However, the rejection regions for the optimal test is a strict subset of the rejection region of \eqref{eqn:rejuc}, as is evident from the conclusion of Theorem \ref{thm:opt1}.
%These regions can be visualized using the graph $\clG$ of Figure \ref{fig:weigtedbip1}.
%The rejection region $\Lambda_R$ of the optimal test of Theorem \ref{thm:opt1} corresponds to all $\underline y$ such that there are two cardinality $K$ matchings on $\clG$ with weight less than or equal to $\widetilde \lambda$ where the weight of a matching is given by \eqref{eqn:matchweight}.
%However, the rejection region $\Lambda_R^{\uc}$ of the unconstrained solution is the set of all $\underline y$ such that some node on the right hand partition has two edges with weight less than or equal to $\widecheck \lambda$.
For large $n$, we have $\widetilde \lambda \approx \widecheck \lambda \approx \lambda$, and thus the sizes of the rejection regions can be significantly different.
The significance can be quantified by comparing the probability of rejection of the two tests.
For large $n$, it is easier to look at the large deviations behavior of these probabilities for which we use rejection exponents.
%, defined analogously to error exponents.
To keep the presentation simple,  in the rest of this section, we focus on the setting in which $M=N=K$.
Furthermore, we assume that the probability distributions in $\clM$ are distinct.% and have complete support on the finite alphabet $\zstate$.

Let $\sigma^i \in \Sym([N]), i \in [J]$ where $J = N!$ denote an enumeration of all possible bijections from $[N]$ onto itself, i.e.,
\[
\sigma^i: [N] \mapsto [N], \quad i \in [J]
\]
represents a unique permutation of $[N]$ for each $i$.
For each $i$ let $\mu^{\sigma^i}$ denote the product distribution $\mu_{\sigma^i(1)} \times \mu_{\sigma^i(2)} \times \ldots \mu_{\sigma^i(N)}$.
It follows, via a straightforward application of Sanov's theorem that the rejection exponents of the optimal test given in Theorem \ref{thm:opt1} and the test $\Lambda^{\uc}$ given in \eqref{eqn:decreguc} and \eqref{eqn:rejuc} can be expressed as follows:
\begin{eqnarray}
%\liminf_{n \to \infty} -\frac{1}{n} \log \Prob_{\clH_\ell}\{\Gamma_Y \in \Lambda_R\}
\lefteqn{\liminf_{n \to \infty} -\frac{1}{n} \log P_\Lambda(\rej/\clH_\ell)}\nonumber\\
%\lefteqn{\liminf_{n \to \infty} -\frac{1}{n} \log \Prob_{\clH_\ell}\{\Gamma_Y \in \Lambda_R\}}\nonumber\\
&=& \left\{ \begin{tabular}{cc}
%$\min_{\sigma \in \Sym([N])}  E_\lambda(\mu_{1} \times \mu_{2} \times \ldots \mu_{N},\mu_{\sigma(1)} \times \mu_{\sigma(2)} \times \ldots \mu_{\sigma(N)})$& if $C^* < \lambda $\\
$\displaystyle \min_{\substack{
    i,j\in[J] \\
   i \neq j
  }}  E_\lambda(\mu^{\sigma^i},\mu^{\sigma^j})$& if $C^* < \lambda $\\
%$\min_{i \neq j}  E_\lambda(\mu^{\sigma^i},\mu^{\sigma^j})$& if $C^* < \lambda $\\
$\infty$ &else
\end{tabular} \right.\label{eqn:rejexpopt}
%-\inf_{\Gamma \in \Lambda_R} \min_{\sigma \in \Sym([N])} D(\Gamma \| \mu_{\sigma_1} \times \mu_{\sigma_2} \times \ldots \mu_{\sigma_N})
\end{eqnarray}
where $C^*=\displaystyle \min_{\substack{
     i,j \in [J] \\
   i \neq j
  }}  C(\mu^{\sigma^i},\mu^{\sigma^j})$,
%$C^*=\min_{\sigma \in \Sym([N])}  C(\mu_{1} \times \mu_{2} \times \ldots \mu_{N},\mu_{\sigma(1)} \times \mu_{\sigma(2)} \times \ldots \mu_{\sigma(N)})$,
and
\begin{eqnarray}
%\liminf_{n \to \infty} -\frac{1}{n} \log \Prob_{\clH_\ell}\{\Gamma_Y \in \Lambda_R^{\uc}\} = \left\{ \begin{tabular}{cc}
\lefteqn{\liminf_{n \to \infty} -\frac{1}{n} \log P_{\Lambda^{\uc}}(\rej/\clH_\ell)}\nonumber\\
&=& \left\{ \begin{tabular}{cc}
$\displaystyle \min_{\substack{
    i,j \in [N] \\
   i \neq j
  }} E_\lambda(\mu_{i}, \mu_{j})$& if $C^{\uc*} < \lambda $\\
$\infty$ &else
\end{tabular} \right.\label{eqn:rejexpoptuc}
%-\inf_{\Gamma \in \Lambda_R} \min_{\sigma \in \Sym([N])} D(\Gamma \| \mu_{\sigma_1} \times \mu_{\sigma_2} \times \ldots \mu_{\sigma_N})
\end{eqnarray}
where $C^{\uc*}=\displaystyle \min_{\substack{
    i,j \in [N] \\
   i \neq j
  }}  C(\mu_{i}, \mu_{j})$.
Thus the important quantities that determine the rejection exponent in the former case are the minimum value of the $E_\lambda$ function and Chernoff information measured between pairs of $\mu^{\sigma^i}$ distributions, and in the latter case, the same functions measured between pairs of $\mu_i$ distributions.
These quantities can differ significantly as we illustrate in Example~\ref{eg:bernoulli} later in the paper.

\subsubsection{Comparison of error probabilities without rejection}
An alternative version of the problem studied in this section is to try to identify an optimal test that does not allow rejection as a test outcome.
When $M \geq N = K$, the problem studied here is just a standard multihypothesis testing problem with $J$ different hypotheses, one corresponding to each permutation of $N$ distributions from the set $\{\mu_1,\mu_2,\ldots, \mu_M\}$.
In this setting, the solution $\what \clH$ given by \eqref{eqn:optmatch1} is in fact the maximum-likelihood solution as shown in \eqref{eqn:HhatML} in Lemma~\ref{lem:ML}.
Also, by applying Lemma~\ref{lem:ML} to the setting in which $N=1$, it follows that the solution $\widehat \sigma$ of \eqref{eqn:ucsigmahat} can be expressed as
\begin{equation}
\what \sigma(i) = \argmax_{j \in [M]} \mu_j(y_i) . \label{eqn:sigmahatML}
%\argmax_{j \in [M]} \prod_{z \in \zstate}\Gamma_{y_i}(z) \log \mu_j(z) . \label{eqn:sigmahatML}
\end{equation}
Thus the solution $\what \sigma$ of \eqref{eqn:ucsigmahat} is the maximum likelihood (ML) solution for the problem studied in this paper when the constraint on the strings is unknown, i.e., it is the ML solution when each string has to be independently classified  to one of the sources without any constraints.
%Similarly, \eqref{eqn:sigmahatML} implies that the solution $\what \sigma$ of \eqref{eqn:ucsigmahat} is the maximum likelihood (ML) solution for the problem studied in this paper when the constraint on the strings is unknown, i.e., it is the ML solution when each string has to be independently classified  to one of the sources without any constraints.
For classical multihypothesis testing problems without rejection, the maximum likelihood solution is known to be asymptotically optimal
%\cite{leajoh97}
% in the Bayesian setting where one assumes that each hypothesis is true with some known prior probability.
%For the non-Bayesian problem, this solution is also optimal
in terms of maximizing the worst-case error exponent \cite{leajoh97} among all hypotheses.
Furthermore, the value of the error exponent is given by the Chernoff information \cite{leajoh97}.
In fact it is straightforward to show that
\begin{equation}
\liminf_{n \to \infty} -\frac{1}{n} \log \Prob_{\clH_\ell}\{\what \clH \neq \clH_\ell\} = C^*, \mbox{ for all } \ell \in [J]. \label{eqn:errexpnorej}
\end{equation}
Furthermore, if $\sigma^\ell \in \Sym([N])$ denotes the permutation function such that $y_i$ is drawn from source $\mu_{\sigma^\ell(i)}$ under hypothesis $\clH_\ell$, then
\begin{equation}
\liminf_{n \to \infty} -\frac{1}{n} \log \Prob_{\clH_\ell}\bigcup_{i=1}^N\{\what \sigma (i) \neq \sigma^\ell(i)\} = C^{\uc*}, \mbox{ for all } \ell \in [J] \label{eqn:errexpnorejuc}
\end{equation}
where $\what \sigma$ is given by \eqref{eqn:ucsigmahat}.

Comparing with \eqref{eqn:rejexpopt} and \eqref{eqn:rejexpoptuc} we see that the the Chernoff informations $C^*$ and $C^{\uc*}$ which determine the error exponents for these tests are equal to the critical values of the error exponent constraints $\lambda$ in the test with rejection, below which the rejection exponent is $\infty$.
%The result of \eqref{eqn:errexpnorej} implies that if $\lambda < C^*$, it is possible to obtain error exponents greater than or equal to $\lambda$ under all hypotheses by using a test that does not have any rejection region at all.
%The result of Theorem~\ref{thm:opt1}, implies that the optimal test given by Theorem~\ref{thm:opt1} must have an empty rejection region for these values of $\lambda$.
%This is indeed true, as verified by \eqref{eqn:rejexpopt}.
In the following example we show that the Chernoff informations $C^*$ and $C^{\uc*}$ for the constrained and unconstrained problems can be significantly different.
Thus the optimal error exponents and rejection exponents in the constrained setting can be significantly higher  than those in the constrained setting.
%In the following example we illustrate the performance enhancement for a simple example with $M=N=2$.

\begin{exmp}\label{eg:bernoulli}
As a simple example, suppose $M=N=K=2$, and $\zstate=\{0,1\}$.
Let $\mu_1$ be given by the Bernoulli distribution with parameter $\half$ and $\mu_2$ a Bernoulli distribution with parameter $\rho$.
In this case $J=N!= 2$ and the two possible permutations are $\sigma^1$ and $\sigma^2$ where $\sigma^1$ is the identity function on $\{1,2\}$ and
\[
\sigma^2(1) = 2 \quad \mbox{ and }\sigma^2(2) = 1.
\]
In this case the distributions $\mu^{\sigma^1} = \mu_1 \times \mu_2$ and $\mu^{\sigma^2} = \mu_2 \times \mu_1$.
These distributions are illustrated in Table~\ref{tab:pmfsilled}.
\begin{table}[ht]
\caption{Probability mass functions illustrated}
\centering
\subtable[PMFs $\mu_1$ and $\mu_2$]{
\begin{tabular}{  | c | c| c| }
\hline
& $0$ & $1$ \\
\hline
$\mu_1$ & $\half$ & $\half$ \\
$\mu_2$ & $1-\rho$ & $\rho$ \\
\hline %inserts single line
\end{tabular}
}
%\label{tab:pmfsimple} % is used to refer this table in the text
%\end{table}
%
%
%\begin{table}[ht]
%\caption{Probability mass functions $\mu^{\sigma^1}$ and $\mu^{\sigma^2}$ illustrated}
%\centering
\subtable[PMFs $\mu^{\sigma^1}$ and $\mu^{\sigma^2}$]{
\begin{tabular}{  |c | c | c | c | c |}
\hline
 &$00$ & $01$ &$10$ & $11$\\
\hline
$\mu^{\sigma^1}$ & $\half (1-\rho)$ & $\half (\rho)$ & $\half (1-\rho)$&$\half (\rho)$\\
$\mu^{\sigma^2}$ & $(1-\rho) \half$ & $(1-\rho) \half $ & $(\rho) \half$&$(\rho) \half$\\
\hline %inserts single line
\end{tabular}
}
\label{tab:pmfsilled} % is used to refer this table in the text
\end{table}

According to the definitions, in this case, the Chernoff informations are given by
\[
C^* = C(\mu^{\sigma^1}, \mu^{\sigma^2}) \quad \mbox{ and } \quad C^{\uc*}= C(\mu_1 ,\mu_2).
\]
These quantities are illustrated as a function of $\rho$ in Figure~\ref{fig:compchern}.
\begin{figure}[h]
\centering
\includegraphics[width=0.8\columnwidth]{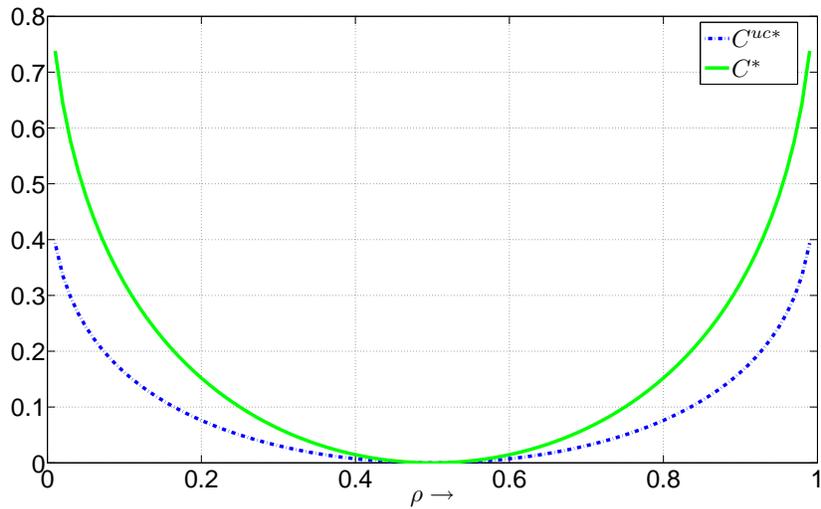}
\caption{Comparison of Chernoff informations $C^{\uc*}$ and $C^*$ for a simple example with $M=N=K=2$ and $\mu_1$ given by the Bernoulli distribution with parameter $\half$ and $\mu_2$ the Bernoulli distribution with parameter $\rho$.} \label{fig:compchern}
\end{figure}
As we see in the figures these quantities are different for all $\rho \neq \half$ and the difference can be significant.

\qed
\end{exmp}

%Thus we conclude from Example~\ref{eg:bernoulli} that the Chernoff information associated with the constrained test can be significantly better than those obtained with the unconstrained test.
%Thus by using the unconstrained solution rather than the constrained solution we get a performance improvement in terms of the error exponent.
Thus we conclude from Example~\ref{eg:bernoulli} and the results of \eqref{eqn:rejexpopt}, \eqref{eqn:rejexpoptuc}, \eqref{eqn:errexpnorej}, and \eqref{eqn:errexpnorejuc}, that by using the unconstrained solution rather than the constrained solution we get a performance improvement in terms of the error exponent.
However, the unconstrained solution has a practical advantage over the optimal solution in terms of the computational complexity of the algorithm for determining the solution, as we elaborate in Section \ref{sec:conc}.

%This improvement can be quantified in terms of rejection probabilities using Proposition~\ref{prop:rejprobcompareknown} or using error probabilities for tests without rejection using Proposition~\ref{prop:errexpnorej}.

%\subsection{Comparison with naive solution}

\section{Optimal Matching with unknown distributions}\label{sec:optunknown}
In this section we solve problem~\Prb{unknown}.
The structure of the solution is very similar to that we obtained in Section~\ref{sec:optknown}.
As before we use the optimality criterion based on error exponents given in \eqref{eqn:criterion}.
In this section we use $\clG$ to denote the graph in Figure~\ref{fig:compbip} with $\clV_1$ and $\clV_2$ respectively representing $\clS_1$ and $\clS_2$, both of which are described in the statement of problem~\Prb{unknown}.
Let $\match_\ell \subset \{1,2,\ldots,M\} \times \{1,2,\ldots,N\}$ with $|\match_\ell| = K$ denote the matching on $\clG$ under hypothesis $\clH_\ell$.
Analogous to our notation in Section~\ref{sec:optknown}, edge $e \in \match_\ell$ can be represented as $e = (e_1,e_2)$ with the understanding that the edge connects $x_{e_1}$ and $y_{e_2}$ in graph $\clG$.
Recall that $\clM$ ($\clN$) represents the probability distributions followed by the $M$ ($N$) sources that produced the sequences in $\clS_1$ ($\clS_2$), and that $\clM \cap \clN = \clK$ with $|\clK| = K$.
%For each $e \in \match_\ell$ let $p_e$ denote the distribution of the source that produced sequences $x_{e_1}$ and $y_{e_2}$. RRRR: Remove confusing notation!!!
%Thus, under hypothesis $\clH_\ell$ we have
%\[
%\clK = \{p_e: e \in \match_\ell\}.
%\]
Thus there are $M-K$ sequences in $\clS_1$ and $N-K$ sequences in $\clS_2$ that are not produced by sources in $\clK$.
%We represent the indices of these sequences under hypothesis $\clH_\ell$ by the following notation
%\[
%I_x^\ell = \{j:\mbox{ No edge in }\match_\ell\mbox{ is incident on }x_j\}
%\]
%and
%\[
%I_y^\ell = \{j:\mbox{ No edge in }\match_\ell\mbox{ is incident on }y_j\}.
%\]
%Furthermore, we let $q^x_i$ denote the probability distribution of the source that produced sequence $x_i$ and $q^y_i$ denote the probability distribution of the source that produced sequence $y_i$.
%Thus under hypothesis $\clH_\ell$, we have
%\[
%\clM = \{q^x_j: j \in I_x^\ell\} \cup \clK
%\]
%and
%\[
%\clN = \{q^y_j: j \in I_y^\ell\} \cup \clK.
%\]

Let
\begin{eqnarray}
D(\clH_\ell)&=& \sum_{(i,j) \in \match_\ell} D(\Gamma_{x_{i}} \| \half({\Gamma_{x_{i}} + \Gamma_{y_j}}) ) \nonumber \\
&& \qquad \qquad +D(\Gamma_{y_j} \| \half({\Gamma_{x_{i}} + \Gamma_{y_j}})).\label{eqn:matchweight2}
\end{eqnarray}
%\begin{eqnarray}
%D(\clH_\ell)&=& \sum_{(j,k) \in \match_\ell} D(\Gamma_{x_{j}} \| \half({\Gamma_{x_{j}} + \Gamma_{y_k}}) )\nonumber\\
%&&   +D(\Gamma_{y_k} \| \half({\Gamma_{x_{j}} + \Gamma_{y_k}})).
%\end{eqnarray}
Consider the estimate for the hypothesis given by
\begin{equation}
\what \clH = \argmin_{\clH_\ell} D(\clH_\ell)\label{eqn:optmatch}
\end{equation}
This test can be interpreted as a minimum weight cardinality-$K$ matching \cite{west2001introduction} on the complete bipartite graph $\clG$ with appropriate weights assigned to the edges in $\clG$.
For $x_i \in \clS_1$ and $y_j \in \clS_2$ let the weight $w_{ij}$ of the edge between them be given by %****notation for edges is not consistent!!!*****
\begin{equation}
w_{ij} = D(\Gamma_{x_i} \| \half({\Gamma_{x_i} + \Gamma_{y_j}}) )+ D(\Gamma_{y_j} \| \half({\Gamma_{x_i} + \Gamma_{y_j}})).\label{eqn:weights}
\end{equation}
As the sequences $x_i$ and $y_j$ have equal lengths, the quantity $\half({\Gamma_{x_i} + \Gamma_{y_j}})$ appearing in \eqref{eqn:weights} can be interpreted as the empirical distribution of the concatenation of $x_i$ and $y_j$.
Thus weight $w_{ij}$ can be interpreted as the sum of two quantities -- the first quantity representing a measure of the difference between sequence $x_i$ and the concatenated sequence, and the second quantity representing a measure of the difference between $y_j$ and the concatenated sequence.
Effectively, $w_{ij}$ can be interpreted as a different distance measure between sequences $x_i$ and $y_j$.
Figure~\ref{fig:weigtedbip} shows the graph $\clG$ with weights added to the edges.

%\begin{figure}[h]
%\centering
%%\psfrag{a}[Bc][Bc]{\small{$x_1$}}
%%\psfrag{b}[Bc][Bc]{\small{$x_2$}}
%%\psfrag{c}[Bc][Bc]{\small{$y_1$}}
%%\psfrag{d}[Bc][Bc]{\small{$y_2$}}
%%\psfrag{e}[Bc][Bc]{\small{$x_K$}}
%%\psfrag{f}[Bc][Bc]{\small{$y_K$}}
%%\psfrag{w}[Bc][Bc]{\small{$w_{11}$}}
%%\psfrag{\clS}[Bc][Bc]{\small{Set $\clS_1$}}
%%\psfrag{Y}[Bc][Bc]{\small{Set $\clS_2$}}
%\includegraphics[width=.6\columnwidth]{FIG/bipartitematch.eps}
%\caption{The solution to the mutiple hypothesis testing problem given in~\eqref{eqn:optmatch} can be obtained by performing a minimum weight bipartite matching with weights given in~\eqref{eqn:weights}.} \label{fig: bipartitematch}
%\end{figure}

\begin{figure}[h]
\centering
\psfrag{s2}{$\clV_1 = \clS_1$}
\psfrag{o2}{$\clV_2 = \clS_2$}

\psfrag{p1}{$x_1$}
\psfrag{p2}{$x_2$}
\psfrag{p3}{$x_i$}
\psfrag{p4}{$x_M$}
\psfrag{vd}{\Huge $\vdots$}
\psfrag{y1}{$y_1$}
\psfrag{y2}{$y_2$}
\psfrag{y3}{$y_j$}
\psfrag{y4}{$y_N$}
\psfrag{wij}{$w_{ij}$}
\includegraphics[width=.4\columnwidth]{weightedbip.eps}
\caption{The weighted complete bipartite graph $\clG$ for problem~\Prb{unknown}. The weight of the edge between the $i$-th vertex in $\clV_1$ and the $j$-th vertex in $\clV_2$ is given by~\eqref{eqn:weights}. The matching corresponding to the hypothesis $\what \clH$ in~\eqref{eqn:optmatch} is given by the minimum weight matching on this graph with cardinality-$K$.} \label{fig:weigtedbip}
\end{figure}

We will now show that a test based on the estimate of $\widehat \clH$ in \eqref{eqn:optmatch} is asymptotically optimal.
%The solution of \eqref{eqn:optmatch} can be justified by the asymptotic optimality properties of a threshold test that uses this statistic.
For proving asymptotic optimality we restrict ourselves to tests that are based only on the empirical distributions of the observations.
Let $\Gamma_{XY}$ denote the collection of empirical distributions:
\[
\Gamma_{XY}:= \left(\Gamma_{x_1},\Gamma_{x_2},\ldots,\Gamma_{x_M},\Gamma_{y_1},\Gamma_{y_2},\ldots,\Gamma_{y_N}\right).
\]
The restriction to tests based on empirical distributions is justified in the asymptotic setting because of the following lemma.
\begin{lemma}\label{lem:typesuff}
Let $\Omega = (\Omega_1, \Omega_2,  \ldots, \Omega_J, \Omega_R)$ be a decision rule based only on the sequences $\{x_1,x_2,\ldots,x_M\}$ and $\{y_1,y_2,\ldots,y_N\}$. Then there exists a decision rule $\Lambda = (\Lambda_1, \Lambda_2, \ldots, \Lambda_J, \Lambda_R)$ based on the sufficient statistics $\Gamma_{XY}$ such that
%\begin{eqnarray}
%&&\limsup_{n \to \infty} \frac{1}{n} \log P_\Lambda(\err/\clH_\ell) \leq  \limsup_{n \to \infty} \frac{1}{n} \log P_\Omega(\err/\clH_\ell), \nonumber\\
%&&\quad \ell=1,2,\ldots,J, \forall \clU \subset \pz \label{eqn:lemclaim}
%\end{eqnarray}
\begin{eqnarray*}
\liminf_{n \to \infty} -\frac{1}{n} \log P_\Lambda(\err/\clH_\ell) &\geq&  \liminf_{n \to \infty} -\frac{1}{n} \log P_\Omega(\err/\clH_\ell),\\%\label{eqn:lemclaim} \\
\liminf_{n \to \infty} -\frac{1}{n} \log P_\Lambda(\rej/\clH_\ell) &\geq&  \liminf_{n \to \infty} -\frac{1}{n} \log P_\Omega(\rej/\clH_\ell). %\label{eqn:rejexporderunknown}
\end{eqnarray*}
for all $\ell \in [J]$ and all $\clU \subset \pz$
%and
%\begin{eqnarray}
%\liminf_{n \to \infty} -\frac{1}{n} \log |\Lambda_R| \geq \liminf_{n \to \infty} -\frac{1}{n} \log |\Omega_R|.\label{eqn:lemrejsize}
%\end{eqnarray}
\qed
\end{lemma}
We provide a proof in the appendix.
Note that $\Lambda_R$ and $\Omega_R$ are finite sets, thus their cardinality is well-defined.

In order to ensure exponential decay of the error probabilities at a prescribed exponential rate, we allow a no-match zone, i.e., we allow a decision in favor of rejecting all the $M$ hypotheses.
For this purpose, we need to identify the hypothesis corresponding to the second minimum weight matching in $\clG$.
Let
\begin{equation}
\wtilde \clH = \argmin_{\clH_\ell \neq \what \clH} D(\clH_\ell)\label{eqn:optmatch3}
\end{equation}
where $\what \clH$ is defined in~\eqref{eqn:optmatch}.
As in Section~\ref{sec:optknown}, the choices of $\what \clH$ and $\wtilde \clH$ have a simple interpretation.% when $M = N = K$.
%However, as the distributions of the sources are now unknown, the interpretations are
based on maximum \emph{generalized likelihoods} as was the case in Lemma~\ref{lem:ML}.
\begin{lemma}\label{lem:GLRunknown}
%If $M = N = K$, t
The selections $\what \clH$ defined in \eqref{eqn:optmatch} and $\wtilde \clH$ defined in \eqref{eqn:optmatch3} can be expressed as
\begin{eqnarray}
\what \clH = \argmax_{\ell \in [J]} \max_{\substack{\clM, \clN \subset \pz\\|\clM \cap \clN| = K} } \Prob_{\clH_\ell} (x_1, \ldots, x_N,y_1, \ldots, y_N), \label{eqn:HhatML2}
\end{eqnarray}
\begin{eqnarray}
\wtilde \clH = \argmax_{\ell \in [J]: \clH_\ell \neq \what \clH} \max_{\substack{\clM, \clN \subset \pz\\|\clM \cap \clN| = K} } \Prob_{\clH_\ell} (x_1, \ldots, x_N,y_1, \ldots, y_N). \label{eqn:HtildeML2}
\end{eqnarray}
%\begin{eqnarray}
%\what \clH &=& \argmax_{\ell \in [J]} \max_{\substack{\clM, \clN \subset \pz\\|\clM \cap \clN| = K} } \Prob_{\clH_\ell} (x_1, x_2, \ldots, x_N,y_1, y_2, \ldots, y_N), \label{eqn:HhatML2}\\
%\wtilde \clH &=& \argmax_{\ell \in [J]: \clH_\ell \neq \what \clH} \max_{\substack{\clM, \clN \subset \pz\\|\clM \cap \clN| = K} } \Prob_{\clH_\ell} (x_1, x_2, \ldots, x_N,y_1, y_2, \ldots, y_N). \label{eqn:HtildeML2}
%\end{eqnarray}
%and the solution $\widehat \sigma$ of \eqref{eqn:ucsigmahat} can be expressed as
%\begin{equation}
%\what \sigma(i) = \argmax_{j \in [M]} \prod_{z \in \zstate}\mu_j(z) y_i(z). \label{eqn:sigmahatML}
%\end{equation}
\qed
\end{lemma}
The above lemma is proved in the appendix.
As in Section~\ref{sec:optknown}, the optimal test with rejection can be stated in terms of $\what \clH$ and $\wtilde \clH$ as described in the following theorem.
%\subsection{Optimality properties}
%We now proceed to characterize the optimality properties enjoyed by the solution of \eqref{eqn:optmatch}.
\begin{theorem}\label{thm:opt}
Let $\Omega = (\Omega_1, \Omega_2,  \ldots, \Omega_J, \Omega_R)$ be a decision rule based on the collection $\Gamma_{XY}$ of empirical distributions
%sequences $\{x_1,\ldots,x_K\}$ and $\{y_1,\ldots,y_K\}$,
such that %for all collections of distributions $p_1,p_2,\ldots,p_K$ from $\pz$ we have
\begin{eqnarray}
P_\Omega(\err/\clH_\ell) \leq 2^{-\lambda n}, \mbox{ for all }\ell \in [J] \label{eqn:errprobconstraint2}
\end{eqnarray}
and for all choices of distributions in $\clK$, $\clM \setminus \clK$, and $\clN \setminus \clK$.

Let $\wtilde \lambda = \lambda - \frac{(M+N)|\zstate| \log(n+1)}{n}$,
\begin{eqnarray*}
\Lambda_\ell =  \{(\underline x, \underline y): D(\wtilde \clH) \geq \wtilde \lambda, \what \clH = \clH_\ell\}, \ell\in [J],
\end{eqnarray*}
and
\[
\Lambda_R = \{(\underline x, \underline y): D(\wtilde \clH) < \wtilde \lambda\}.
\]
Then
\begin{eqnarray}
\liminf_{n \to \infty} -\frac{1}{n} \log P_\Lambda(\err/\clH_\ell) \geq  \lambda, \quad \ell \in [J], \forall \clU \subset \pz \label{eqn:thmclaim1}
\end{eqnarray}
%\begin{eqnarray}
%&&\limsup_{n \to \infty} \frac{1}{n} \log P_\Lambda(\err/\clH_\ell) \leq - \lambda, \nonumber\\
%&&\ell=1,2,\ldots,J, \forall \clU \subset \pz \label{eqn:thmclaim1}
%\end{eqnarray}
and
\begin{eqnarray}
\Lambda_R \subset \Omega_R. \label{eqn:thmclaim2}
\end{eqnarray}
\qed
\end{theorem}
%We provide a proof to the theorem in the appendix.
%From the definition of $D(H_\ell)$ in \eqref{eqn:matchweight2} it is clear that the decision regions $\Lambda_\ell$'s and $\Lambda_R$ proposed in Theorem \ref{thm:opt} depends on the sequences in $(\underline x, \underline y)$ only through $\Gamma_{XY}$.
%It is also clear that the solution given by Theorem \ref{thm:opt} and that proposed in \eqref{eqn:optmatch} only arises due to the rejection region.
%As seen in the statement of the theorem, the no-match decision is made if the second minimum weight matching also has low weight.
%Due to this rejection region, we can now guarantee exponential decay of the error probabilities with a given exponent under the various hypotheses.
%The guarantee on the error probability follows from a result similar to that of Lemma~\ref{lem:ldptwoemp} which quantifies the probability of large deviations of empirical distributions of a pair of strings drawn from the same distribution.
%We note that the offset between $\wtilde \lambda$ and $\lambda$ is just $(M+N)$ times the offset in the exponent appearing in the first inequality of \eqref{eqn:typeineq} which bounds the probability of observing a type.
%This offset is introduced to ensure that the condition \eqref{eqn:thmclaim2} is satisfied.
%The detailed arguments are in the proof.
%
%
%
%=====
We provide a proof to the theorem in the appendix.
From the definition of $D(\clH_\ell)$ in \eqref{eqn:matchweight2} it is clear that the decision regions $\Lambda_\ell$'s and $\Lambda_R$ proposed in Theorem \ref{thm:opt} depends on the sequences in $\underline x$ and $\underline y$ only through $\Gamma_{XY}$.
From the condition of \eqref{eqn:thmclaim2} it is obvious that the probability of rejection of the decision rules under the decision rule $\Lambda$ is lower than the probability of rejection under the decision rule $\Omega$.
Thus the optimality result implies that the test $\Lambda$ has lower rejection probabilities, as defined in \eqref{eqn:rejprob2}, than any test $\Omega$ that satisfies an exponential decay of error probabilities as in \eqref{eqn:errprobconstraint2}.
For $N=1$, this result is similar to that given by Gutman in \cite[Thm 2]{gut89}.
However, the rejection condition in this solution is different from that provided by Gutman and can be interpreted as the condition under which the second lowest weight matching has a weight below a threshold.

The test can be explained in words as follows.
First identify the hypothesis corresponding to the minimum weight matching and the second minimum weight matching of cardinality $K$ in $\clG$.
Accept the former hypothesis if the weight corresponding to the latter exceeds the threshold $\wtilde \lambda$, and reject all hypotheses if the threshold is not exceeded.
In the case of $M=N=K$, it follows via Lemma~\ref{lem:GLRunknown} that the optimal choice of the hypothesis is given by the maximum generalized likelihood hypothesis, and that a no-match decision is selected when the second highest generalized likelihood exceeds a threshold.
In other words, this test leads to a rejection if the observations can be  well-explained by two or more hypotheses.

We note that the threshold $\wtilde \lambda$ appearing in the definition of $\Lambda_R$  satisfies $\wtilde \lambda \to \lambda$ as $n \to \infty$.
Using a generalization of Lemma~\ref{lem:ldptwoemp} we show in the proof that the choice of decision regions ensures that the error-exponent constraint of \eqref{eqn:thmclaim1} is satisfied.
We also observe that the offset between $\wtilde \lambda$ and $\lambda$ is just $(M+N)$ times the offset in the exponent appearing in the first inequality of \eqref{eqn:typeineq} which bounds the probability of observing a type.
This offset is introduced to ensure that the condition \eqref{eqn:thmclaim2} is satisfied.
The detailed arguments are in the proof.

%======
%
%In the case of $M=N=K$, it follows via Lemma~\ref{lem:GLRunknown} that the optimal choice of the hypothesis is given by the maximum generalized likelihood hypothesis, and that a no-match decision is selected when the second highest generalized likelihood exceeds a threshold.

%\newpage

\subsection{Comparison with the unconstrained problem}\label{sec:compunknown}
As in Section~\ref{sec:optknown}, it is interesting to compare the solution of Theorem \ref{thm:opt} with the solution that one would have to use without the prior knowledge that the strings in $\clS_2$ (also $\clS_1$) are produced by distinct sources.
We follow the same steps as in Section~\ref{sec:compknown}.
Assume $M=N=K$.
In the absence of prior knowledge, a reasonable strategy is to try to sequentially match each string in $\clS_2$ to some string in $\clS_1$.
For each $i \in [N]$ define
\begin{equation}
\widehat \sigma(i) = \argmin_{j\in [M]} D(\Gamma_{x_{j}} \| \half({\Gamma_{x_{j}} + \Gamma_{y_i}}) ) +D(\Gamma_{y_i} \| \half({\Gamma_{x_{j}} + \Gamma_{y_i}})). \label{eqn:ucsigmahat2}
\end{equation}
The function $\widehat \sigma$ maps every string in $\clS_2$ to some string in $\clS_1$.
If $\widehat \sigma$ is a one-to-one function, then it corresponds to a valid hypothesis for the matching problem.
We call this hypothesis $\clH^{\widehat \sigma}$.
If ${\widehat \sigma}$ is not a one-to-one function, or if $\clH^{\widehat \sigma}$ does not correspond to the true hypothesis, then the strings are not correctly matched and hence in this case an error occurs.
Furthermore, in order to satisfy the error exponent constraint, one is forced to reject whenever the individual hypothesis test on any of the $N$ strings in $\clS_2$ leads to a rejection.
%For each $i$ we identify the weight of the second lightest edge in $\clG$ incident at $s_i$.
For $i \in [N]$, let
\[
\widetilde w_i = \min_{j\in [M]\setminus{{\widehat \sigma}(i)}} D(\Gamma_{x_{j}} \| \half({\Gamma_{x_{j}} + \Gamma_{y_i}}) ) +D(\Gamma_{y_i} \| \half({\Gamma_{x_{j}} + \Gamma_{y_i}})).
\]
The solution to the overall problem is now given by
\begin{eqnarray}
\Lambda_\ell^{\uc} =  \{(\underline x, \underline y): \min_{i \in [N]} \widetilde w_i \geq \widecheck \lambda, \clH^{\widehat \sigma}= \clH_\ell\}, \ell \in [J]\label{eqn:decreguc2}
\end{eqnarray}
where the superscript of $\uc$ indicates that the solution is unconstrained and
\begin{eqnarray}
\Lambda_R^{\uc} = \{(\underline x, \underline y): \min_{i \in [N]} \widetilde w_i < \widecheck \lambda\} \label{eqn:rejuc2}% \cup \{\Gamma_Y: {\widehat \sigma} \mbox{ is not a one-one-function}\}
\end{eqnarray}
where $\widecheck \lambda = \lambda - \frac{(N+1)|\zstate| \log(n+1)}{n}$, is the optimal choice for the threshold obtained from Theorem \ref{thm:opt} when the second set of strings $\clS_2$ is a singleton.
%It is to be noted that in this case $\bigcup_{\ell \in [J]} \Lambda_\ell^{\uc} \bigcup \Lambda_R^{\uc} \neq \mathbf{Z}^1$ because it is possible that ${\widehat \sigma}$ might not be a permutation.
The probability of error of this solution is given by
\begin{eqnarray*}
\lefteqn{P_{\Lambda^{\uc}}(\err|\clH_\ell)} \\ &=& \Prob_{\clH_\ell}\{y_i \mbox{ is incorrectly matched for some }i \in [N]\}\\
&\leq& \sum_{i=1}^N \Prob_{\clH_\ell}\{y_i \mbox{ is incorrectly matched}\}
\end{eqnarray*}
where the inequality follows via the union bound.
By the result of Theorem \ref{thm:opt}, each term in the above summation decays exponentially in $n$ with exponent $\lambda$ and thus we have
\[
\liminf_{n \to \infty} -\frac{1}{n} \log P_{\Lambda^{\uc}}(\err|\clH_\ell) \geq \lambda.
\]
Thus this solution meets the same error exponent constraint as the optimal solution of Theorem \ref{thm:opt} that one can use when the constraint on the strings is known a priori.
However, the rejection regions for the optimal test is a strict subset of the rejection region of \eqref{eqn:rejuc2}, as is evident from the conclusion of Theorem \ref{thm:opt}.
These regions can be visualized using the graph $\clG$ of Figure \ref{fig:weigtedbip}.
The rejection region $\Lambda_R$ of the optimal test of Theorem \ref{thm:opt} corresponds to all sequences $(\underline x, \underline y)$ such that there are two cardinality $K$ matchings on $\clG$ with weight less than or equal to $\widetilde \lambda$ where the weight of a matching is given by \eqref{eqn:matchweight}.
However, the rejection region $\Lambda_R^{\uc}$ of the unconstrained solution is the set of all $(\underline x, \underline y)$  such that some node on the right hand partition has two edges with weight less than or equal to $\widecheck \lambda$.
For large $n$, we have $\widetilde \lambda \approx \widecheck \lambda \approx \lambda$, and thus the sizes of the rejection regions can be significantly different.
As in Section~\ref{sec:compknownrej}, it is possible to use Sanov's theorem to quantify the significance by comparing the rejection exponents of the two tests.
Similarly it is possible to compare error exponents obtained by tests that do not allow rejection.
However, the analysis is much more involved as in the problem~\Prb{unknown} we now have two strings from each source, and hence analytical expressions for the rejection exponents are difficult to obtain.
We therefore avoid the details in this paper.

%\newpage

%%%%%%%%%%%%%%%%%%%%%%%%%%%%%%%%%%%%%%%%%%%%%%%%%%%%%%%%%%%%%%%%%
%%%%%%%%%%%%%%%%%%%%%%%%%%%%%%%%%%%%%%%%%%%%%%%%%%%%%%%%%%%%%%%%%
%%%%%%%%%%%%%%%%%%%%%%%%%%%%%%%%%%%%%%%%%%%%%%%%%%%%%%%%%%%%%%%%%

\section{Practical aspects, extensions and conclusions}\label{sec:conc}
We proposed asymptotically optimal solutions to two hypothesis testing problems that seek to match unlabeled sequences of observations to labeled source distributions or training sequences.
Under the constraint that the observed sequences are drawn from distinct sources, the structure of the optimal solution is significantly different from the unconstrained solution, and can lead to significant improvement in performance as we saw in Sections~\ref{sec:compknown} and~\ref{sec:compunknown}.
%We solved the problems under the requirement that the error probabilities must decay exponentially at a pre-specified rate.
%For this, we allowed a no-match decision in which all possible hypotheses are rejected.
%It may be possible to generalize the ideas in this paper to other optimality criteria which do not allow rejection.

An important practical aspect is that of the complexity of the algorithms required to identify the optimal solutions of Theorem~\ref{thm:opt1} and~\ref{thm:opt}.
The unconstrained solutions of \eqref{eqn:ucsigmahat} and \eqref{eqn:ucsigmahat2} are straightforward to identify because these can be obtained by sequentially matching each string in $\clS$ or $\clS_2$ to one of the $M$ sources in $\clM$ or one of the $M$ strings in $\clS_1$.
This leads to a time-complexity of $O(MN)$.
The optimal (constrained) solutions are in general more complex to identify as a combinatorial optimization problem has to be solved to identify $\what H$ and $\wtilde H$ defined in \eqref{eqn:optmatch1}, \eqref{eqn:optmatch}, \eqref{eqn:optmatch31} and \eqref{eqn:optmatch3}.
Nevertheless, these solutions can be identified by solving minimum weight bipartite matching problems on the graphs $\clG$ which can be executed efficiently in polynomial-time.

The first step in implementing the solutions to problems~\Prb{known} and~\Prb{unknown} is to identify the estimates of \eqref{eqn:optmatch1} and \eqref{eqn:optmatch}.
As discussed earlier, the task of identifying these estimates is equivalent to solving a minimum weight cardinality-$K$ matching problem on a weighted complete bipartite graph.
If $M = N = K$, then this problem can be solved using the Hungarian algorithm \cite{kuh55}, which has a time-complexity of $O(N^3)$ (see \cite{ramtar12} and references therein).
When $M \neq N$, the Hungarian algorithm can be adapted to run in $O(MNK)$ as detailed in \cite{ramtar12b}.
%$O(MNK + K^2 \log (\min\{M,N\}))$ as detailed in \cite{ramtar12b}.
Thus the complexity of this algorithm is roughly $K$ times more than that of the naive unconstrained algorithm.
This problem can also be solved using a polynomial time algorithm based on the theory of matroids (see, e.g., \cite[Ch. 8]{coocunpulsch11}).
In practice the complexity can often be reduced significantly.
For instance, in the solution of \eqref{eqn:optmatch1}, if some empirical distribution $\Gamma_{y_j}$ is not absolutely continuous with respect to some $\mu_i$, then the edge connecting the corresponding vertices in $\clG$ can be removed, as it will never be selected in the minimum weight matching.
The same step can be performed if the empirical distributions $\Gamma_{x_i}$ and $\Gamma_{y_j}$ have disjoint supports in the solution of \eqref{eqn:optmatch}.
If the number of remaining edges in the graph is $E$, then, when $M=N=K$, the Hungarian algorithm can be adapted to run with a complexity of $O(EN + N^2 \log N)$ \cite{fretar87} and when $M \neq N$, with a complexity of $O(EK +
K^2 \log (\min\{M,N\}))$ \cite{ramtar12b}.
Once the matchings $\what \clH$ of \eqref{eqn:optmatch1} and \eqref{eqn:optmatch} are identified, the matchings corresponding to \eqref{eqn:optmatch31} and \eqref{eqn:optmatch3} can also be identified in polynomial time.
A naive algorithm for this would be to sequentially repeat the same algorithms on the graphs obtained by removing edges appearing in $\what \clH$ one at a time from the graph $\clG$.
The minimum weight matching obtained in all repetitions would correspond to the minimum weight matching on the original graph $\clG$ that is not identical to $\what \clH$.
In many practical applications this step is unnecessary as rejecting all hypotheses is not acceptable.
In such cases one can use the estimates of \eqref{eqn:optmatch1} and \eqref{eqn:optmatch}.
A practical application of such a solution and experimental evaluation of the method is reported in \cite{unnnai13} where $M=N=K \approx 1500$ and for $M=N=K \approx 47000$  in \cite{naiunnthivet14}.

The proposed solution can be extended in many directions.
An important generalization is with respect to the requirement that all sequences have the same sample size.
In practice, this is often not the case.
For example in problem~\Prb{known}, it might be the case that each string $y_i$ has a length $n_i = \alpha_i n$ with $\alpha_i \geq 1$.
In such a case it might be possible to extend the result of Theorem~\ref{thm:opt1} by adapting the definitions of $D(\clH_\ell)$ in \eqref{eqn:matchweight} to
\[
D(\clH_\ell) = \sum_{(i,j) \in \match_\ell} \alpha_j D(\Gamma_{y_j} \| \mu_i)
\]
and redefining the threshold $\wtilde \lambda$ in the statement of the theorem to $\wtilde \lambda = \lambda - \frac{|\zstate| \sum_{i=1}^N\log (n_i + 1)}{n}$.
%\rho(n)$ where $\rho(n) =
With this definition the optimality result of Theorem~\ref{thm:opt1} is expected to hold.
Moreover the maximum likelihood interpretation of Lemma~\ref{lem:ML} continues to hold for $M \geq N = K$.
Alternatively, if all $n_i \geq n$ and all $n_i$ are approximately equal, then the test proposed in Theorem~\ref{thm:opt1} can still be used, and the probability of error and probability of rejection are expected to be lower than those expected when $n_i = n$ for all $i$.

Similarly, the solution to problem~\Prb{unknown} can be generalized to the scenario in which the sequences have distinct lengths by following Gutman \cite{gut89}.
Let $n^x_i = \alpha^x_i n$ denote the length of sequence $x_i \in \clS_1$ and $n^y_j = \alpha^y_j n$ denote the length of sequence $y_j \in \clS_2$ with $\alpha^x_i \geq 1$ and $\alpha^y_j \geq 1$ for all $i,j$.
Then the definition of $D(\clH_\ell)$ in \eqref{eqn:matchweight2} can be changed to
\begin{eqnarray*}
D(\clH_\ell)&=& \sum_{(i,j) \in \match_\ell} \frac{n^x_i}{n}D\left(\Gamma_{x_{i}} \| \frac{n^x_i \Gamma_{x_{i}} + n^y_j\Gamma_{y_j}}{n^x_i+n^y_j}\right)  \\
&& \qquad+ \frac{n^y_j}{n}D\left(\Gamma_{y_j} \| \frac{n^x_i\Gamma_{x_{i}} + n^y_j\Gamma_{y_j}}{n^x_i+n^y_j}\right)
\end{eqnarray*}
and the threshold $\wtilde \lambda$ in the statement of Theorem~\ref{thm:opt1} can be changed to 
\[
\wtilde \lambda = \lambda - \frac{|\zstate| \left(\sum_{i=1}^M\log (n^x_i + 1) + \sum_{j=1}^N\log (n^y_j + 1)\right)}{n}.
\]
With this definition the optimality result of Theorem~\ref{thm:opt} is expected to hold.
It can be noted that $\frac{n^x_i \Gamma_{x_{i}} + n^y_j\Gamma_{y_j}}{n^x_i+n^y_j} = \Gamma_{t_{ij}}$ where $t_{ij}$ is the concatenation of $x_i$ and $y_j$.

%
%
%The approach of \cite{gut89} can be used to generalize the solution to problem~\Prb{unknown} to the scenario in which the training sequences in $\clS_1$ have lengths equal to $n$ while the observation sequences in $\clS_2$ have lengths equal to $\widetilde n$.
%%We remark that the test proposed in Theorem \ref{thm:opt} and the optimality result can be generalized to the scenario in which the training sequences in $\clS_1$ have lengths equal to $n$ while the observation sequences in $\clS_2$ have lengths equal to $\tilde n$, and $\tilde n$ is linear in $n$.
%In this setting, the weight $w_{ij}$ appearing in \eqref{eqn:weights} has to be replaced by
%\[
%D(\Gamma_{x_i} \| \Gamma_{t_{ij}})+ \frac{\wtilde n}{n}D(\Gamma_{y_j} \| \Gamma_{t_{ij}})
%\]
%where $t_{ij}$ is the concatenation of $x_i$ and $y_j$ and $\Gamma_{t_{ij}}$ is the empirical distribution of $t_{ij}$.
%With this modification the test proposed in Theorem \ref{thm:opt} continues to be optimal for an appropriately chosen threshold $\widetilde \lambda$ provided $\wtilde n$ is linear in $n$.

Throughout this paper we focused on source probability distributions supported on a finite alphabet $\zstate$.
It might be possible to extend some of these results to probability distributions on continuous alphabets.
For the problem with known sources studied in Section~\ref{sec:optknown}, we know from Lemma~\ref{lem:ML}, that the choices $\what H$ and $\wtilde H$ correspond to the maximum-likelihood hypothesis, and the second most likely hypothesis.
Hence, even for continuous alphabets, these hypotheses can be identified using standard techniques \cite{leh05}.
However, we recall that the optimal test of Theorem~\ref{thm:opt1} requires us to compare $D(\wtilde H)$ to a threshold.
For continuous distributions the empirical distributions are in general not absolutely continuous with respect to the true distributions and thus $D$ is always $\infty$.
Thus the definition of the decision regions for the optimal test have to be modified by replacing the Kullback Leibler divergence with some appropriately defined function of the log-likelihood function and by setting the thresholds intelligently.
A potential approach is to adapt the method proposed for binary hypothesis testing in \cite{han00} to multiple hypothesis.
The analysis of the error exponents and rejection exponents in such continuous alphabet problems are typically performed using Cramer's theorem  rather than Sanov's theorem.
The error exponent result of \eqref{eqn:errexpnorej} is expected to continue to hold for a test that always decides in favor of $\what H$ without rejection.

%See \cite{han00}, \cite{alacherac04}

The results of Section~\ref{sec:optunknown} are more difficult to generalize to source distributions on continuous alphabets, because, in general, the empirical distributions of all sequences are expected to have mutually disjoint supports.
However, if the source distributions are constrained to lie in some parametric family, for example, an exponential family \cite{leh05} such as the class of Gaussian distributions of unknown means and variances equal to unity, it might be possible to identify optimal procedures via the maximum generalized likelihood interpretation of Lemma~\ref{lem:GLRunknown}.
This idea of restricting to finite dimensional parametric families is similar to the dimensionality reduction approach prescribed in \cite{unnhuameysurvee11} for universal hypothesis testing.
Theses ideas are also useful in applications in which the alphabet size $|\zstate|$ is large.
As described in \cite{unnhuameysurvee11}, test statistics for hypothesis testing problems on large alphabets suffer from large variance for moderate sequence lengths, and thus lead to poor error probability performances.
Dimensionality reduction techniques like those proposed in \cite{unnhuameysurvee11} are an effective technique to address these concerns.

%An interesting observation can be made in the limiting regime in which $n \gg \wtilde n$.
%In this case the $x_i$ is much longer than $y_j$ and hence $\Gamma_{t_{ij}} \approx \Gamma_{x_{i}}$.
%Thus in this limiting regime $w_{ij} \approx

The solutions proposed in this paper for i.i.d. sources can also be easily extended to finite memory Markov sources on finite alphabets following the approach in \cite{gut89}.
Furthermore, it is possible to study the weak-convergence behavior of the test statistics of \eqref{eqn:optmatch1} and \eqref{eqn:optmatch} following the method outlined in \cite{unnhua13}.
Using such an analysis it is possible to estimate the error probabilities for finite sample sizes.

\section*{Acknowledgments}
The author thanks the anonymous reviewers for several helpful suggestions and Farid Movahedi Naini for helpful discussions.
%Richard Timsit and Yves Despond for providing us with the EPFL Wi-Fi connection logs, and Elio Abi Karam for his help in the experimental evaluation.
This research was supported by ERC Advanced Investigators Grant: Sparse Sampling: Theory, Algorithms and Applications SPARSAM no 247006.
% Authors also thank Ola!
% this work is supported by grant ...

%%%%%%%%%%%%%%%%%%%%%%%%%%%%%%%%%%%%%%%%%%%%%%%%%%%%%%%%%%%%%%%%%
%%%%%%%%%%%%%%%%%%%%%%%%%%%%%%%%%%%%%%%%%%%%%%%%%%%%%%%%%%%%%%%%%
%%%%%%%%%%%%%%%%%%%%%%%%%%%%%%%%%%%%%%%%%%%%%%%%%%%%%%%%%%%%%%%%%

\appendix
For proving the various results we need some new notation and a few lemmas.
For any sequence $s \in \zstate^n$ we use $T_s$ to denote the \emph{type class} of $s$, i.e., the set of all sequences of length $n$ with the same empirical distribution as $s$.
The following lemmas are well known.
For proofs see \cite{covtho06}.
The first lemma below is just a restatement of Lemma~\ref{lem:typeprobnew}.
\begin{lemma}\label{lem:typeprob}
For every $p \in \pz$ and every $s \in \zstate^n$,
\[
\frac{1}{(n+1)^{|\zstate|}} 2^{-n D(\Gamma_s \| p)} \leq \Prob_p(T_s) \leq 2^{-n D(\Gamma_s \| p)}
\]
where $\Prob_p$ denotes the probability measure when all observations in $s$ are drawn i.i.d. according to law $p$. \qed
\end{lemma}
\begin{lemma}\label{lem:probandentropy}
For any sequence $s \in \zstate^n$  and any $\nu \in \pz$ we have
\[
\nu(s) \leq 2^{-n H(\Gamma_s)}.
\]
\qed
\end{lemma}
The following lemma is easy to see.
\begin{lemma}\label{lem:sumexpentropy}
For finite set $\zstate$, we have
%Let $ z = (z_1,z_2,\ldots,z_n)$. Then
\[
\sum_{s \in \zstate^n} 2^{-n(H(\Gamma_{s}))} \leq (n+1)^{|\zstate|}.
\]
\end{lemma}
\begin{proof}
Let $\clP_n$ denote the set of all types with denominator $n$.
Let $T(P)$ be the set of sequences in $\zstate^n$ with type $P$.
We have
\begin{eqnarray}
\sum_{s \in \zstate^n} 2^{-n(H(\Gamma_{s}))} &=& \sum_{P \in \clP_n} |T(P)| 2^{-n H(P)}\\
&\stackrel{(a)}{\leq}& \sum_{P \in \clP_n} 1 = |\clP_n| \stackrel{(b)}{\leq} (n+1)^{|\zstate|}
\end{eqnarray}
where (a) and (b) follow from \cite[Ch. 11]{covtho06}.
\end{proof}
The following lemma is also required for some proofs.
%Then we have the following lemma.
\begin{lemma}\label{lem:KLofunion}
For $\mu_1, \mu_2, \ldots, \mu_N \in \clP(\zstate)$ let $\mu^{\mathsf{prod}} := \mu_1 \times \mu_2 \times \ldots \mu_N$ denote the product distribution.
For $\pi \in \clP(\zstate^N)$, let $\breve\pi$ denote the product distribution obtained from the marginals of $\pi$, i.e., $\breve\pi = \pi_{1,.}\times \pi_{2,.} \times \ldots \pi_{N,.}$, where $\pi_{k,.}$ denote the marginal distribution of $\pi$ with respect to the $k$-th component.
Then we have
\[
D(\breve\pi \| \mu^{\mathsf{prod}}) \leq D( \pi \| \mu^{\mathsf{prod}}).
\]
\end{lemma}
\begin{proof}
We know that
\[
D( \pi \| \mu^{\mathsf{prod}}) = E_{X_1,X_2,\ldots,X_N} \frac{\log \pi(X_1,X_2,\ldots,X_N)}{\prod_{i \in [N]}\mu_i(X_i)}
\]
where $(X_1,X_2,\ldots,X_N)$ has joint distribution $\pi$.
Simplifying we have
\begin{eqnarray*}
\lefteqn{D( \pi \| \mu^{\mathsf{prod}})}\\
 &=& E_{X_1,X_2,\ldots,X_N} \log \left(\frac{\pi(X_1,X_2,\ldots,X_N)}{\prod_{i \in [N]}\pi_{i,.}(X_i)}\right) \\
 &&\qquad +\log \left(\frac{\prod_{i \in [N]}\pi_{i,.}(X_i)}{\prod_{i \in [N]}\mu_i(X_i)}\right)\\
&=& E_{X_1,X_2,\ldots,X_N} \log \left(\frac{\pi(X_1,X_2,\ldots,X_N)}{\prod_{i \in [N]}\pi_{i,.}(X_i)}\right) \\
&&\qquad+\sum_{i \in [N]} D(\pi_{i,.} \| \mu_i)\\
&=& \sum_{i\in[N]} H(X_i) - H(X_1,X_2,\ldots,X_N) + \sum_{i \in [N]} D(\pi_{i,.} \| \mu_i)\\
&\geq& \sum_{i \in [N]} D(\pi_{i,.} \| \mu_i)
= D(\breve\pi \| \mu^{\mathsf{prod}}).
%\log \left(\frac{\prod_{i \in [N]}\pi_{i,.}(X_i)}{\prod_{i \in [N]}\mu_i(X_i)}\right).
\end{eqnarray*}
where $H(.)$ denotes Shannon entropy and the inequality follows a well known information theoretic inequality between the joint Shannon entropy of random variable and the sum of their individual Shannon entropies \cite{covtho06}.
\end{proof}

\subsection{Proof of Lemma~\ref{lem:ldptwoemp}}
Let $A = \{(x,y): x \in \ystate^n, y \in \ystate^n, \mbox{ and }D(\Gamma_{x} \| \half(\Gamma_{x}  + \Gamma_{y} ))  + D(\Gamma_{y} \| \half(\Gamma_{x}  + \Gamma_{y} )) > \lambda\}$.
Then we have
\begin{eqnarray*}
\lefteqn{\Prob \{ ({y_1}, {y_2}) \in A\}}\\ &=& \sum_{(x,y) \in A} \nu(x) \nu(y)\\
&\stackrel{(a)}{\leq}& \sum_{(x,y) \in A} 2^{-2n H(\half(\Gamma_x + \Gamma_y))}\\
&=& \sum_{(x,y) \in A} 2^{-n (H(\Gamma_x) + H(\Gamma_y)  )}\\
&&\qquad 2^{-n (D(\Gamma_x \| \half(\Gamma_x + \Gamma_y))  + D(\Gamma_y \| \half(\Gamma_x + \Gamma_y)) )}\\
&\stackrel{(b)}{\leq}& \sum_{(x,y) \in A} 2^{-n (H(\Gamma_x) + H(\Gamma_y) + \lambda )}\\
&\leq& 2^{-n  \lambda} \sum_{x \in \ystate} 2^{-n H(\Gamma_x) } \sum_{y \in \ystate} 2^{-n H(\Gamma_y)}\\
&\stackrel{(c)}{\leq}& 2^{-n  \lambda}(n+1)^{2|\ystate|}
\end{eqnarray*}
where (a) follows from Lemma~\ref{lem:probandentropy} applied to a concatenation of $x$ and $y$, (b) from the definition of $A$, and (c) from Lemma~\ref{lem:sumexpentropy}.
Thus
\begin{eqnarray*}
\lim_{n \to \infty} - \frac{1}{n} \log \Prob \{ ({y_1}, {y_2}) \in A\} &\geq& \lambda.
\end{eqnarray*}

\subsection{Proof of Lemma \ref{lem:typesuff1}}
Consider an arbitrary tuplet of sequences $(y_1,y_2,\ldots,y_N)$.
Let $T=(T_{y_1},\ldots,T_{y_N})$ denote the joint type-class of all the sequences, i.e., it is the set of all tuplets of sequences with the same joint type as $(y_1,y_2,\ldots,y_N)$:
\[
T = \{(z_1,\ldots,z_N):  z_i \subset \zstate^n \mbox{ and } \Gamma_{z_i} = \Gamma_{y_i} \mbox{ for all } i \in [N]\}.
\]
Any $(y_1',y_2',\ldots,y_N')\in T$ belongs to exactly one of the sets $\Omega_1,\Omega_2,\ldots,\Omega_J,\Omega_R$.
We modify the decision rule $\Omega$ as follows.
For any joint type $T$ we let $\Lambda_\ell$ include $T$ if $\Omega_\ell$ contains the most number of the sequences of $T$, for $\ell \in \{1,2,\ldots,J, R\}$.
In case of ties we break them arbitrarily and include $T$ in exactly one of the $\Lambda_\ell$'s.

%By construction we have for any joint type $T \subset \Lambda_R$
%\[
%|\Omega_R| \geq \frac{1}{J+1}|T|.
%\]
%Hence
%\begin{eqnarray*}
%|\Lambda_R| &=& \sum_{T \subset \Lambda_R} |T|\\
%&\leq&  \sum_{T \subset \Lambda_R} (J+1)|\Omega_R|\\
%%&\leq&  |\Omega_R|(1+(J+1) \tau_n )
%&\leq&  |\Omega_R|(J+1) \tau_n
%\end{eqnarray*}
%where $\tau_n$ represents the number of joint types of length $n$.
%%*****RRRRR is $1+$ required in equation above?*****
%Since $\frac{\log \tau_n}{n} \to 0$ \cite{covtho06} we have \eqref{eqn:lemrejsize1}.

Let $q^y_i$ denote the probability distribution of the source that produced sequence $y_i$ under hypothesis $\clH_\ell$.
For any hypothesis $\clH_\ell$ with $\ell \in [J]$ and any joint type $T \subset \Lambda_k$ with $k \in [J] \cup \{R\}$ we have by Lemma \ref{lem:typeprob} and definition of $\Lambda_\ell$:
\begin{eqnarray}
\Prob_{\clH_\ell} \{\Omega_k\} &\geq& \Prob_{\clH_\ell} \{\Omega_k \cap T\} \geq \frac{1}{J+1}\Prob_{\clH_\ell} \{T\} \nonumber\\
&\stackrel{(a)}{\geq}& \frac{2^{-n \left(\delta(n)+ \sum_{j=1}^N D(\Gamma_{y_j} \| q_j^y)\right)}}{J+1} \label{eqn:prooftemp}
\end{eqnarray}
where (a) follows via the first inequality in the statement of Lemma \ref{lem:typeprob} with $\delta(n) = \frac{N |\zstate| \log(n+1)}{n}$.
Combining the above result along with the definition of $\Lambda_\ell$ and Lemma \ref{lem:typeprob}, we have
\begin{eqnarray*}
\Prob_{\clH_\ell} \{\Lambda_k\} &=& \sum_{T: T \subset \Lambda_k}\Prob_{\clH_\ell} \{T\}\\
&\stackrel{(a)}{\leq}& \sum_{T: T \subset \Lambda_k} 2^{-n \left(\sum_{j=1}^N D(\Gamma_{y_j} \| q_j^y)\right)}\\
&\stackrel{(b)}{\leq}& \sum_{T: T \subset \Lambda_k} 2^{n \delta(n)} (J+1) \Prob_{\clH_\ell} \{\Omega_k\}\\
&\leq& \tau_n 2^{ n \delta(n)}(J+1)\Prob_{\clH_\ell} \{\Omega_k\}
\end{eqnarray*}
where (a) follows via the second inequality in Lemma \ref{lem:typeprob} and (b) via \eqref{eqn:prooftemp}.
The quantity $\tau_n$ represents the number of joint types of length $n$. 
Since $\frac{\log \tau_n}{n} \to 0$ \cite{covtho06} and $\delta(n) \to 0$ we obtain the inequality relations claimed in the lemma by choosing $k \in [J]$ and $k = R$.

\subsection{Proof of Lemma~\ref{lem:ML}}
%\begin{proof}

Under $\clH_\ell$ let
\[
I_y^\ell = \{j:\mbox{ No edge in }\match_\ell\mbox{ is incident on }y_j\}.
\]
Also let $q^y_i$ denote the probability distribution of the source that produced sequence $y_i$.
We have
\begin{eqnarray*}
\lefteqn{\argmax_{\ell \in [J]} \max_{\substack{\clN \subset \pz\\\clM \cap \clN = \clK} }\Prob_{\clH_\ell} (y_1, y_2, \ldots, y_N)} \\
&=& \argmax_{\ell \in [J]} \sum_{(i,j) \in \match_\ell} \log {\prod_{k=1}^n\mu_i(y_j(k))} \\
&&\qquad+ \sum_{j \in I_y^\ell} \max_{q^y_j \in \pz}\log {\prod_{k=1}^n q^y_j(y_j(k))}\\
&\stackrel{(a)}{=}&  \argmax_{\ell \in [J]} \sum_{(i,j) \in \match_\ell} \left(\sum_{z \in \zstate} n \Gamma_{y_j}(z) \log(\mu_i(z)) \right)  \\
&&\qquad + \sum_{j \in I_y^\ell} \log {\prod_{k=1}^n \Gamma_{y_j} (y_j(k))}\\
&=& \argmax_{\ell \in [J]} \sum_{(i,j) \in \match_\ell} \left(-H(\Gamma_{y_j})-D(\Gamma_{y_j} \| \mu_i)\right) \\
&&\qquad- \sum_{j \in I_y^\ell} H(\Gamma_{y_j})\\
&=& \argmin_{\ell \in [J]} \sum_{(i,j) \in \match_\ell} D(\Gamma_{y_j} \| \mu_i)+ \sum_{j \in [N]} H(\Gamma_{y_j}) \\
&=&\argmin_{\ell \in [J]} D(\clH_\ell)%\\
%&\stackrel{(a)}{=}&
%&=&\argmin_{\ell \in [J]} D(\clH_\ell)
\end{eqnarray*}
where (a) follows from the fact that the likelihood of a string is maximized by the empirical distribution.

\subsection{Proof of Theorem \ref{thm:opt1}}
As before, let
\[
I_y^\ell = \{j:\mbox{ No edge in }\match_\ell\mbox{ is incident on }y_j\}
\]
denote the indices of the $N-K$ sequences in $\clS$ that are produced by sources in $\clN \setminus \clK$.
Similarly, let
\[
I_\mu^\ell = \{j:\mbox{ No edge in }\match_\ell\mbox{ is incident on }\mu_j\}.
\]
denote the indices of the $M-K$ sources in $\clM \setminus \clK$.
%, whose indices are collectively represented as
%
%
%Similarly, there are $N-K$ sequences in $\clS$ that are produced by sources in $\clN \setminus \clK$.
%The indices of these sequences are represented via the following notation:
%\[
%I_y^\ell = \{j:\mbox{ No edge in }\match_\ell\mbox{ is incident on }y_j\}.
%\]
We continue to use $q^y_i$ to denote the probability distribution of the source that produced sequence $y_i$.
%Thus under hypothesis $\clH_\ell$, we have
%\[
%\clM = \{\mu_j: j \in I_\mu^\ell\} \cup \clK
%\]
%and
%\[
%\clN = \{q^y_j: j \in I_y^\ell\} \cup \clK.
%\]
Let
\[
\wtilde \Lambda_\ell =   \{\underline y: D(\clH_\ell) \geq \wtilde \lambda\}, \ell \in [J].
\]
The probability of error of decision rule $\Lambda$ under hypothesis $\clH_\ell$ is given by
\[
P_\Lambda(\err/ \clH_\ell) = \Prob_{\clH_\ell}\left\{\underline y \in \displaystyle\bigcup_{\substack{k=1\\k \neq \ell}}^J\Lambda_k \right\} \leq \Prob_{\clH_\ell} (\wtilde \Lambda_\ell).
\]
%By Sanov's theorem \cite{demzei98a}, we have
%\begin{eqnarray*}
%\liminf_{n \to \infty} -\frac{1}{n} \log \Prob_{\clH_\ell} (\wtilde \Lambda_\ell) &=& \inf \{D(\nu \| )\}
%\end{eqnarray*}
We observe that $D(\clH_\ell)$ in the definition of $\wtilde \Lambda_\ell$ is a sum of the Kullback Leibler divergences between the empirical distributions of each $y_i$ and some $\mu_j$.
The empirical distributions of each $y_i$ can be interpreted as the marginal of a joint empirical distribution of $\underline y$ interpreted as a sequence of length $n$ drawn from $\zstate^N$.
We also note that $\wtilde \lambda \to \lambda$ and $n \to \infty$.
The result of \eqref{eqn:thmclaim11} follows directly by applying Sanov's theorem \cite{demzei98a} combined with the conclusion of Lemma~\ref{lem:KLofunion}.

%***RRRR copying ****

For proving \eqref{eqn:thmclaim21} we observe that for any test based on empirical distributions, we have
\begin{eqnarray*}
2^{-\lambda n} &\geq& P_\Omega(\err/\clH_\ell) = \sum_{\cup_{k\neq \ell} \Omega_k}  \prod_{j=1}^N q_j^y(y_j)
\end{eqnarray*}
where we use $q_j^y(s)$ to denote the probability that sequence $s$ was generated i.i.d. under law $q_j^y$.
Simplifying further we have,
\begin{eqnarray*}
2^{-\lambda n} &\geq& \sum_{\cup_{k\neq \ell} \Omega_k}  \prod_{i \in I_y^\ell} q_i^y(y_i) \prod_{j \notin I_y^\ell} q_j^y(y_j) \nonumber \\
%&& \qquad\prod_{e \in \match_\ell} q_{e_1}^x(x_{e_1}) q_{e_1}^x(y_{e_2})  \nonumber \\
%&\stackrel{(a)}{\geq}& \sum_{T \subset \cup_{k\neq \ell} \Omega_k} %2^{-n \sum_{i \in I_x^\ell} \left(D(\Gamma_{x_i} \| q_i^x) + \delta(n)\right)  }\\
%%&&
%2^{-n \sum_{j \in I_y^\ell} \left(D(\Gamma_{y_j} \| q_j^y) + \delta(n)\right)  }\\
%&& 2^{-n \sum_{j \notin I_y^\ell} \left(D(\Gamma_{y_j} \| q_j^y) + \delta(n)\right)  }\\
&\stackrel{(a)}{\geq}& \sum_{T \subset \cup_{k\neq \ell} \Omega_k} %2^{-n \sum_{i \in I_x^\ell} \left(D(\Gamma_{x_i} \| q_i^x) + \delta(n)\right)  }\\
%&&
2^{-n \sum_{i \in I_y^\ell} \left(D(\Gamma_{y_i} \| q_i^y) + \delta(n)\right)  }\\
&& \qquad 2^{-n \sum_{j \notin I_y^\ell} \left(D(\Gamma_{y_j} \| q_j^y) + \delta(n)\right)  }\\
%2^{-n \sum_{e \in \match_\ell} \left( D(\Gamma_{x_{e_1}} \| q_{e_1}^x) + D(\Gamma_{y_{e_2}} \| q_{e_2}^y) +2\delta(n)\right)}\\
%&\geq&
%2^{-n \sum_{j \in I_y^\ell} \left(D(\Gamma_{y'_j} \| q_j^y) + \delta(n)\right)  }\\
%&& 2^{-n \sum_{j \notin I_y^\ell} \left(D(\Gamma_{y'_j} \| q_j^y) + \delta(n)\right)  }
&\geq&
2^{-n \sum_{i \in I_y^\ell} \left(D(\Gamma_{y'_i} \| q_i^y) + \delta(n)\right)  } \\
&& \qquad 2^{-n \sum_{j \notin I_y^\ell} \left(D(\Gamma_{y'_j} \| q_j^y) + \delta(n)\right)  }
%2^{-n \sum_{i \in I_x^\ell} \left(D(\Gamma_{x_i'} \| q_i^x) + \delta(n)\right)  }\\
%&&2^{-n \sum_{j \in I_y^\ell} \left(D(\Gamma_{y_j'} \| q_j^y) + \delta(n)\right)  }\\
%&&2^{-n \sum_{e \in \match_\ell} \left( D(\Gamma_{x_{e_1}'} \| q_{e_1}^x) + D(\Gamma_{y_{e_2}'} \| q_{e_1}^x) +2\delta(n)\right)}
\end{eqnarray*}
where (a) follows from Lemma \ref{lem:typeprob} with $T = (T_{y_1},\ldots,T_{y_N})$ and $\delta(n) = \frac{|\zstate| \log(n+1)}{n}$, and $({y_1'},{y_2'},\ldots,{y_N'}) \in \cup_{k\neq \ell} \Omega_k$, and  all distributions in %$\clU \subset \pz$ are arbitrary.
$\clN \setminus \clK \subset \pz$ are arbitrary.
If we specifically choose $\clN \setminus \clK $ such that $q_j^y = \Gamma_{y_j'}$ for all $j \in  I_y^\ell$ we get
\begin{eqnarray*}
\lambda &\leq& \sum_{j \notin I_y^\ell} \left( D(\Gamma_{y_{j}'} \| q_j^y) \right) + N\delta(n)
\end{eqnarray*}
which further implies that
%Defining $\rho(n) = 2 K \delta(n)$ results in
\begin{eqnarray}
\cup_{j\neq \ell} \Omega_j \subset \wtilde \Lambda_\ell.\label{eqn:omincl1}
\end{eqnarray}
Now let
\[
\what \Lambda_\ell := \cap_{j \neq \ell} \wtilde \Lambda_j.
\]
Hence,
\[
\cup_\ell \Lambda_\ell = \{\underline y: D(\wtilde \clH) \geq \wtilde \lambda\} = \cup_\ell \what \Lambda_\ell.
\]
Combining with \eqref{eqn:omincl1} we get
\[
\what \Lambda_\ell = \cap_{j \neq \ell} \wtilde \Lambda_j \supset \cap_{j \neq \ell} \cup_{k\neq j} \Omega_k \supset \Omega_\ell
\]
and thus
\[
\Lambda_R^c = \cup_\ell \Lambda_\ell = \cup_\ell \what \Lambda_\ell \supset \cup_\ell \Omega_\ell = \Omega_R^c.
\]
Hence
\[
\Lambda_R \subset \Omega_R.
\]

\subsection{Proof of Lemma \ref{lem:typesuff}}
This proof is very similar to that of Lemma \ref{lem:typesuff1}.
Let $(x_1,x_2,\ldots,x_M,y_1,y_2,\ldots,y_N)$ be an arbitrary tuplet of sequences.
Let $T=(T_{x_1},\ldots,T_{x_M},T_{y_1},\ldots,T_{y_N})$ denote the joint type-class of all the sequences, defined similarly to the definition in the proof of Lemma \ref{lem:typesuff1} as
\begin{eqnarray*}
T &=& \left\{(w_1,\ldots,w_M,z_1,\ldots,z_N):  w_i, z_j \subset \zstate^n, \Gamma_{w_i} = \Gamma_{x_i} \right. \\
&&\qquad \left.  \mbox{ and }\Gamma_{z_j} = \Gamma_{y_j} \mbox{ for all } i \in [M], j \in [N]\right\}.
\end{eqnarray*}
Any $(x_1',x_2',\ldots,x_M',y_1',y_2',\ldots,y_N')\in T$ belongs to exactly one of the sets $\Omega_1,\Omega_2,\ldots,\Omega_J,\Omega_R$.
We modify the decision rule $\Omega$ as follows.
For any joint type $T$ we let $\Lambda_\ell$ include $T$ if $\Omega_\ell$ contains the most number of the sequences of $T$, for $\ell \in \{1,2,\ldots,J, R\}$.
In case of ties we break them arbitrarily and include $T$ in exactly one of the $\Lambda_\ell$'s.

%By construction we have for any joint type $T \subset \Lambda_R$
%\[
%|\Omega_R| \geq \frac{1}{J+1}|T|.
%\]
%Moreover, we have
%\begin{eqnarray*}
%|\Lambda_R| &=& \sum_{T \subset \Lambda_R} |T|\\
%&\leq&  \sum_{T \subset \Lambda_R} (J+1)|\Omega_R|\\
%%&\leq&  |\Omega_R|(1+(J+1) \tau_n )
%&\leq&  |\Omega_R|(J+1) \tau_n
%\end{eqnarray*}
%where $\tau_n$ represents the number of joint types of length $n$.
%%*****RRRRR is $1+$ required in equation above?*****
%Since $\frac{\log \tau_n}{n} \to 0$ \cite{covtho06} we have \eqref{eqn:lemrejsize}.

Under hypothesis $\clH_\ell$, let $q^x_i$ denote the probability distribution of the source that produced sequence $x_i$ and $q^y_i$ denote the probability distribution of the source that produced sequence $y_i$.
For any hypothesis $\ell \in [J]$ and any joint type $T \subset \Lambda_k$ with $k \in [J] \cup \{R\}$ we have by Lemma \ref{lem:typeprob} and definition of $\Lambda_\ell$:
\begin{eqnarray*}
\Prob_{\clH_\ell} \{\Omega_k\} &\geq& \Prob_{\clH_\ell} \{\Omega_k \cap T\} \geq \frac{1}{J+1}\Prob_{\clH_\ell} \{T\} \\
&\geq& \frac{2^{-n \left(\delta(n)+\sum_{i=1}^M D(\Gamma_{x_i} \| q_i^x)+ \sum_{j=1}^N D(\Gamma_{y_j} \| q_j^y)\right)}}{J+1}
\end{eqnarray*}
where $\delta(n) = \frac{(M+N) |\zstate| \log(n+1)}{n}$.
Combining the above result along with the definition of $\Lambda_k$ and Lemma \ref{lem:typeprob}, we have
\begin{eqnarray*}
\Prob_{\clH_\ell} \{\Lambda_k\} &=& \sum_{T \subset \Lambda_k}\Prob_{\clH_\ell} \{T\}\\
&\leq& \sum_{T \subset \Lambda_k} 2^{-n \left(\sum_{i=1}^M D(\Gamma_{x_i} \| q_i^x)+ \sum_{j=1}^N D(\Gamma_{y_j} \| q_j^y)\right)}\\
&\leq& \sum_{T \subset \Lambda_k} 2^{n \delta(n)} (J+1) \Prob_{\clH_\ell} \{\Omega_k\}\\
&\leq& \tau_n 2^{ n \delta(n)}(J+1)\Prob_{\clH_\ell} \{\Omega_k\}
\end{eqnarray*}
where $\tau_n$ represents the number of joint types of length $n$.
Since $\frac{\log \tau_n}{n} \to 0$ \cite{covtho06} and $\delta(n) \to 0$ the results follow by choosing $k \in [J]$ and $k = R$.
\subsection{Proof of Lemma~\ref{lem:GLRunknown}}
%\begin{proof}
We know that there are $M-K$ sequences in $\clS_1$ and $N-K$ sequences in $\clS_2$ that are not produced by sources in $\clK$.
We represent the indices of these sequences under hypothesis $\clH_\ell$ by the following notation
\begin{equation}
I_x^\ell = \{j:\mbox{ No edge in }\match_\ell\mbox{ is incident on }x_j\}\label{eqn:ixl}
\end{equation}
and
\begin{equation}
I_y^\ell = \{j:\mbox{ No edge in }\match_\ell\mbox{ is incident on }y_j\}\label{eqn:iyl}.
\end{equation}
Furthermore, we let $q^x_i$ denote the probability distribution of the source that produced sequence $x_i$ and $q^y_i$ denote the probability distribution of the source that produced sequence $y_i$.

We first observe that if $x, y \in \zstate^n$ are two length $n$ strings drawn under the same distribution from $\pz$, then the maximum likelihood distribution that produced it is given by $\half(\Gamma_{x} + \Gamma_{y})$, the empirical distribution of the concatenated string.
In other words
\begin{eqnarray}
\argmax_{\mu \in \pz} \mu(x) \mu(y) %&=& 
%&=& \argmin_{\mu \in \pz} D(\half(\Gamma_{x} + \Gamma_{y})\| \mu)\nonumber \\
&=& \half(\Gamma_{x} + \Gamma_{y})\label{eqn:MLempirical}
\end{eqnarray}
%This is because
%\begin{eqnarray}
%\argmax_{\mu \in \pz} \mu(x) \mu(y) &=& \argmax_{\mu \in \pz} \prod_{z \in \zstate}\mu(z)^{n_x(z) + n_y(z)} \nonumber \\
%%&=& \argmax_{\mu \in \pz} \prod_{z \in \zstate}\mu(z)^{n_x(z) + n_y(z)}\nonumber \\
%&=& \argmax_{\mu \in \pz} \log \prod_{z \in \zstate}\mu(z)^{{n}(\Gamma_{x}(z) + \Gamma_{y}(z))}\nonumber \\
%&=& \argmax_{\mu \in \pz} \sum_{z \in \zstate}({n}(\Gamma_{x}(z) + \Gamma_{y}(z))) \log \mu(z)\nonumber \\
%&=& \argmax_{\mu \in \pz} \sum_{z \in \zstate}\half(\Gamma_{x}(z) + \Gamma_{y}(z)) \log \mu(z)\nonumber \\
%&=& \argmin_{\mu \in \pz} \sum_{z \in \zstate}-\half(\Gamma_{x}(z) + \Gamma_{y}(z)) \log \mu(z)\nonumber \\
%&=& \argmin_{\mu \in \pz} H(\half(\Gamma_{x} + \Gamma_{y})) + D(\half(\Gamma_{x} + \Gamma_{y})\| \mu)\nonumber \\
%&=& \argmin_{\mu \in \pz} D(\half(\Gamma_{x} + \Gamma_{y})\| \mu)\nonumber \\
%&=& \half(\Gamma_{x} + \Gamma_{y})\label{eqn:MLempirical}
%\end{eqnarray}
%where $n_x(z)$ (respectively, $n_y(z)$) denotes the number of times $z$ occurs in string $x$ ($y$).

%====New try =======

\noindent Now we have
\begin{eqnarray*}
\lefteqn{  \log \Prob_{\clH_\ell} (x_1, x_2, \ldots, x_M, y_1, y_2, \ldots, y_N)}\\
&=&  \sum_{(i,j) \in \match_\ell} \sum_{z \in \zstate} \log {(q^x_i(z))^{{{n}(\Gamma_{x_i}(z) + \Gamma_{y_j}(z)) }}} \\
&& \qquad +\sum_{ i \in I_x^\ell} \sum_{z \in \zstate} \log {(q^x_i(z))^{n \Gamma_{x_i}(z) }  } \\
&&\qquad +\sum_{ j \in I_y^\ell} \sum_{z \in \zstate} \log {(q^y_j(z))^{ n \Gamma_{y_j}(z) }}
\end{eqnarray*}
By \eqref{eqn:MLempirical}
\begin{eqnarray*}
\lefteqn{\max_{\substack{\clM, \clN \subset \pz\\|\clM \cap \clN| = K} } \log \Prob_{\clH_\ell} (x_1, x_2, \ldots, x_M, y_1, y_2, \ldots, y_N)}\\
%&=&   \sum_{(i,j) \in \match_\ell} \sum_{z \in \zstate} \log {\left(\frac{\Gamma_{x_i}(z) + \Gamma_{y_j}(z)}{2}\right)^{{n_{x_i}(z) + n_{y_j}(z)} }} \\
%&& +\sum_{ i \in I_x^\ell}  \sum_{z \in \zstate} \log {(\Gamma_{x_i}(z) )^{{n_{x_i}(z) } }} +\sum_{ j \in I_y^\ell}  \sum_{z \in \zstate} \log {(\Gamma_{y_j}(z) )^{{n_{y_j}(z) } }}\\
&=& \sum_{(i,j) \in \match_\ell} \sum_{z \in \zstate} \log {\left(\frac{\Gamma_{x_i}(z) + \Gamma_{y_j}(z)}{2}\right)^{{n}(\Gamma_{x_i}(z) + \Gamma_{y_j}(z)) }}\\
&&\qquad +\sum_{ i \in I_x^\ell}  \sum_{z \in \zstate} \log {(\Gamma_{x_i}(z) )^{n \Gamma_{x_i}(z) }}\\
&& \qquad  +\sum_{ j \in I_y^\ell}  \sum_{z \in \zstate} \log {(\Gamma_{y_j}(z) )^{n \Gamma_{y_j}(z) }}\\
%&=& \argmax_{\ell \in [J]} \sum_{(i,j) \in \match_\ell} D(\Gamma_{x_i} \| \half(\Gamma_{x_i} + \Gamma_{y_j})) + D(\Gamma_{y_j} \| \half(\Gamma_{x_i} + \Gamma_{y_j})) + H(\Gamma_{x_i}) + H(\Gamma_{y_j})\\
\end{eqnarray*}
Hence
\begin{eqnarray*}
\lefteqn{\argmax_{\ell \in [J]}\max_{\substack{\clM, \clN \subset \pz\\|\clM \cap \clN| = K} } \log \Prob_{\clH_\ell} (x_1, \ldots, x_M, y_1, \ldots, y_N)}\\
%&=&   \sum_{(i,j) \in \match_\ell} \sum_{z \in \zstate} \log {\left(\frac{\Gamma_{x_i}(z) + \Gamma_{y_j}(z)}{2}\right)^{{n_{x_i}(z) + n_{y_j}(z)} }} \\
%&& +\sum_{ i \in I_x^\ell}  \sum_{z \in \zstate} \log {(\Gamma_{x_i}(z) )^{{n_{x_i}(z) } }} +\sum_{ j \in I_y^\ell}  \sum_{z \in \zstate} \log {(\Gamma_{y_j}(z) )^{{n_{y_j}(z) } }}\\
&=& \argmin_{\ell \in [J]} \sum_{(i,j) \in \match_\ell} \left\{ D(\Gamma_{x_i} \| \half(\Gamma_{x_i} + \Gamma_{y_j})) \right.\\
&&\qquad \left.+ D(\Gamma_{y_j} \| \half(\Gamma_{x_i} + \Gamma_{y_j}))\right\} \\
&&\qquad - \sum_{i \in [M]} H(\Gamma_{x_i}) - \sum_{j \in [N]} H(\Gamma_{y_j}) \\
&{=}& \argmin_{\ell \in [J]} \sum_{(i,j) \in \match_\ell} \left\{D(\Gamma_{x_i} \| \half(\Gamma_{x_i} + \Gamma_{y_j})) \right.\\
&&\qquad \left.+ D(\Gamma_{y_j} \| \half(\Gamma_{x_i} + \Gamma_{y_j}))\right\}
\end{eqnarray*}
where the last step follows from the fact that the sum of the entropy terms is equal for all hypotheses.
The conclusion of \eqref{eqn:HhatML2} follows directly, and that of \eqref{eqn:HtildeML2} by a similar argument.

\subsection{Proof of Theorem \ref{thm:opt}}
%We know that there are $M-K$ sequences in $\clS_1$ and $N-K$ sequences in $\clS_2$ that are not produced by sources in $\clK$.
We use the notation of $I_x^\ell$ and $I_y^\ell$ introduced in \eqref{eqn:ixl} and \eqref{eqn:iyl}.
%represent the indices of these sequences under hypothesis $\clH_\ell$ by the following notation
%\[
%I_x^\ell = \{j:\mbox{ No edge in }\match_\ell\mbox{ is incident on }x_j\}
%\]
%and
%\[
%I_y^\ell = \{j:\mbox{ No edge in }\match_\ell\mbox{ is incident on }y_j\}.
%\]
Furthermore, as before let $q^x_i$ (respectively $q^y_i$) denote the probability distribution of the source that produced sequence $x_i$ ($y_i$).
%and $q^y_i$ denote %the probability distribution of the source that produced sequence .
%Thus under hypothesis $\clH_\ell$, we have
%\[
%\clM = \{q^x_j: j \in I_x^\ell\} \cup \clK
%\]
%and
%\[
%\clN = \{q^y_j: j \in I_y^\ell\} \cup \clK.
%\]

This proof is very similar to that of Theorem \ref{thm:opt1}.
Define
\[
\wtilde \Lambda_\ell =   \{(\underline x, \underline y): D(\clH_\ell) \geq \wtilde \lambda\}, \ell \in [J].
\]
Clearly,
\[
\Lambda_j \subset \wtilde \Lambda_\ell \mbox{ for all }j \neq \ell
\]
and hence
\[
\cup_{j\neq \ell} \Lambda_j \subset\cup_{j\neq \ell} \left(\cap_{k \neq j}\wtilde \Lambda_k \right) \subset \wtilde \Lambda_\ell.
\]
Therefore,
\begin{eqnarray}
P_\Lambda(\err/\clH_\ell) &=& \sum_{\cup_{k\neq \ell} \Lambda_k } \prod_{i \in I_x^\ell} q_i^x(x_i)  \prod_{j \in I_y^\ell} q_j^y(y_j) \nonumber \\
&& \qquad  \prod_{e \in \match_\ell} q_{e_1}^x(x_{e_1}) q_{e_2}^y(y_{e_2}) \nonumber \\
%&\stackrel{(a)}{=}&
%\sum_{\cup_{k\neq \ell} \Lambda_k } \prod_{i \in I_x^\ell} q_i^x(x_i) \prod_{j \in I_y^\ell} q_j^y(y_j) \nonumber \\
%&& \qquad\prod_{e \in \match_\ell} q_{e_1}^x(x_{e_1}) q_{e_1}^x(y_{e_2})  \nonumber \\
&\stackrel{(a)}{\leq}&  \sum_{\wtilde \Lambda_\ell} \prod_{i \in I_x^\ell} q_i^x(x_i) \prod_{j \in I_y^\ell} q_j^y(y_j) \nonumber \\
&& \qquad\prod_{e \in \match_\ell} q_{e_1}^x(x_{e_1}) q_{e_1}^x(y_{e_2}) \label{eqn:eqncrit}
\end{eqnarray}
where (a) follows from the fact that under $\clH_\ell$ for all $e \in \match_\ell$ we have $q_{e_1}^x = q_{e_2}^y$.
If $\match_\ell = \{e^1,e^2,\ldots,e^K\}$, let
%\[
%\Gamma_{XY,\ell} = (\Gamma_{x_{e_1^1}},\Gamma_{x_{e_1^2}},\ldots,\Gamma_{x_{e_1^K}},\Gamma_{y_{e_2^1}},\Gamma_{y_{e_2^2}},\ldots,\Gamma_{y_{e_2^K}}).
%\]
%Let
%$\overline \Lambda_\ell$ denote
%\[
%\overline \Lambda_\ell := \{\Gamma_{XY,\ell}: (\underline x, \underline y) \in \wtilde \Lambda_\ell\}.
%\]
\[
\overline \Lambda_\ell := \{({x_{e_1^1}},{x_{e_1^2}},\ldots,{x_{e_1^K}},{y_{e_2^1}},{y_{e_2^2}},\ldots,{y_{e_2^K}}): (\underline x, \underline y) \in \wtilde \Lambda_\ell\}.
\]
From the definition of $\wtilde \Lambda_\ell$ it is evident that for a fixed matching $\match_\ell$, the membership of $(\underline x, \underline y)$ in $\wtilde \Lambda_\ell$ depends only on $\overline \Lambda_\ell$.
%$\Gamma_{XY,\ell}$.
Hence \eqref{eqn:eqncrit} becomes:
\begin{eqnarray*}
P_\Lambda(\err/\clH_\ell) &\leq&  \sum_{\overline \Lambda_\ell} \prod_{e \in \match_\ell} q_{e_1}^x(x_{e_1}) q_{e_1}^x(y_{e_2})  \\
&\stackrel{(b)}{\leq}& \sum_{\overline \Lambda_\ell} \prod_{e \in \match_\ell} 2^{-2n H(\half({\Gamma_{x_{e_1}} + \Gamma_{y_{e_2}}}))} \\
&=& \sum_{\overline \Lambda_\ell} 2^{-2n \sum_{e \in \match_\ell} H(\half({\Gamma_{x_{e_1}} + \Gamma_{y_{e_2}}}))} 
\end{eqnarray*}
where (b) follows from Lemma~\ref{lem:probandentropy}.
Note that
\begin{eqnarray*}
2H(\half({\Gamma_{x_{e_1}} + \Gamma_{y_{e_2}}})) &=& H(\Gamma_{x_{e_1}}) + H(\Gamma_{y_{e_2}}) \\
&&+ D(\Gamma_{x_{e_1}} \| \half({\Gamma_{x_{e_1}} + \Gamma_{y_{e_2}}}) )\\
&&+ D(\Gamma_{y_{e_2}} \| \half({\Gamma_{x_{e_1}} + \Gamma_{y_{e_2}}}))
\end{eqnarray*}
Thus 
\begin{eqnarray*}
P_\Lambda(\err/\clH_\ell) 
%&\leq&  
%\sum_{\overline \Lambda_\ell} 2^{-n \sum_{e \in \match_\ell} \left(H(\Gamma_{x_{e_1}}) + H(\Gamma_{y_{e_2}}) + D(\Gamma_{x_{e_1}} \| \half({\Gamma_{x_{e_1}} + \Gamma_{y_{e_2}}}) )+ D(\Gamma_{y_{e_2}} \| \half({\Gamma_{x_{e_1}} + \Gamma_{y_{e_2}}}))\right)} \\
%%&=& \sum_{\overline \Lambda_\ell} 2^{-n \sum_{e \in \match_\ell} \left(H(\Gamma_{x_{e_1}}) + H(\Gamma_{y_{e_2}}) + \right.}\nonumber\\
%%&& ^{\left.D(\Gamma_{x_{e_1}} \| \half({\Gamma_{x_{e_1}} + \Gamma_{y_{e_2}}}) )+\right.}\nonumber\\
%%&&^{\left.D(\Gamma_{y_{e_2}} \| \half({\Gamma_{x_{e_1}} + \Gamma_{y_{e_2}}}))\right)} \\
&\leq& \sum_{\overline \Lambda_\ell} 2^{-n \wtilde \lambda -n \sum_{e \in \match_\ell} \left(H(\Gamma_{x_{e_1}}) + H(\Gamma_{y_{e_2}})  \right)} \\
&\stackrel{(c)}{\leq}& 2^{-n \wtilde \lambda} \left((n+1)^{|\zstate|}\right)^{2 |\match_\ell| } \\
&=& 2^{-n \wtilde \lambda} \left(n+1\right)^{2K|\zstate|} \\
%&\leq& 2^{-n \wtilde \lambda} \sum_{\bfZ} 2^{-n \sum_{e \in \match_\ell} \left(H(\Gamma_{x_{e_1}}) + H(\Gamma_{y_{e_2}})\right)} \\
&\leq& 2^{-n (\lambda + O(\frac{\log n}{n}))}
\end{eqnarray*}
where (c) follows from Lemma \ref{lem:sumexpentropy}.
This proves \eqref{eqn:thmclaim1}.
This proof can be interpreted as an extension of Lemma~\ref{lem:ldptwoemp} to $K$ pairs of empirical distributions.

For proving \eqref{eqn:thmclaim2} we observe that for any test based on empirical distributions, we have
\begin{eqnarray*}
2^{-\lambda n} &\geq& P_\Omega(\err/\clH_\ell) \\
&=& \sum_{\cup_{k\neq \ell} \Omega_k}  \prod_{i=1}^M q_i^x(x_i) \prod_{j=1}^N q_j^y(y_j)\\
&=& \sum_{\cup_{k\neq \ell} \Omega_k}  \prod_{i \in I_x^\ell} q_i^x(x_i) \prod_{j \in I_y^\ell} q_j^y(y_j) \nonumber \\
&& \qquad\prod_{e \in \match_\ell} q_{e_1}^x(x_{e_1}) q_{e_1}^x(y_{e_2})  \nonumber \\
&\stackrel{(a)}{\geq}&  \sum_{T \subset \cup_{k\neq \ell} \Omega_k} 2^{-n \sum_{i \in I_x^\ell} \left(D(\Gamma_{x_i} \| q_i^x) + \delta(n)\right)  }\\
&&2^{-n \sum_{j \in I_y^\ell} \left(D(\Gamma_{y_j} \| q_j^y) + \delta(n)\right)  }\\
&&2^{-n \sum_{e \in \match_\ell} \left( D(\Gamma_{x_{e_1}} \| q_{e_1}^x) + D(\Gamma_{y_{e_2}} \| q_{e_1}^x) +2\delta(n)\right)}\\
&\geq& 2^{-n \sum_{i \in I_x^\ell} \left(D(\Gamma_{x_i'} \| q_i^x) + \delta(n)\right)  }\\
&&2^{-n \sum_{j \in I_y^\ell} \left(D(\Gamma_{y_j'} \| q_j^y) + \delta(n)\right)  }\\
&&2^{-n \sum_{e \in \match_\ell} \left( D(\Gamma_{x_{e_1}'} \| q_{e_1}^x) + D(\Gamma_{y_{e_2}'} \| q_{e_1}^x) +2\delta(n)\right)}
\end{eqnarray*}
where (a) follows from Lemma \ref{lem:typeprob} with $T = (T_{x_1},\ldots,T_{x_M},T_{y_1},\ldots,T_{y_N})$ and $\delta(n) = \frac{|\zstate| \log(n+1)}{n}$, and $({x_1'},{x_2'},\ldots,{x_M'},{y_1'},{y_2'},\ldots,{y_N'}) \in \cup_{k\neq \ell} \Omega_k$ and  all distributions in $\clU \subset \pz$ are arbitrary.
If we specifically choose $\clU$ such that $ q_{e_1}^x = \half (\Gamma_{x_{e_1}'} + \Gamma_{y_{e_2}'})$ for all $e \in \match_\ell$, and $q_i^x = \Gamma_{x_i'}$ for all $i \in  I_x^\ell$ and $q_j^y = \Gamma_{y_j'}$ for all $j \in  I_y^\ell$ we get
\begin{eqnarray*}
\lambda &\leq& \sum_{e \in \match_\ell} \left( D(\Gamma_{x_{e_1}'} \| \half (\Gamma_{x_{e_1}'} + \Gamma_{y_{e_2}'})) \right.\nonumber\\
&& \left.+   D(\Gamma_{y_{e_2}'} \| \half (\Gamma_{x_{e_1}'} + \Gamma_{y_{e_2}'})) \right) + (M+N)\delta(n)
\end{eqnarray*}
which further implies that
%Defining $\rho(n) = 2 K \delta(n)$ results in
\begin{eqnarray}
\cup_{j\neq \ell} \Omega_j \subset \wtilde \Lambda_\ell.\label{eqn:omincl}
\end{eqnarray}
Now let
\[
\what \Lambda_\ell := \cap_{j \neq \ell} \wtilde \Lambda_j.
\]
Hence,
\[
\cup_\ell \Lambda_\ell = \{(\underline x, \underline y): D(\wtilde \clH) \geq \wtilde \lambda\} = \cup_\ell \what \Lambda_\ell.
\]
Combining with \eqref{eqn:omincl} we get
\[
\what \Lambda_\ell = \cap_{j \neq \ell} \wtilde \Lambda_j \supset \cap_{j \neq \ell} \cup_{k\neq j} \Omega_k \supset \Omega_\ell
\]
and thus
\[
\Lambda_R^c = \cup_\ell \Lambda_\ell = \cup_\ell \what \Lambda_\ell \supset \cup_\ell \Omega_\ell = \Omega_R^c.
\]
Hence
\[
\Lambda_R \subset \Omega_R.
\]
\bibliographystyle{IEEEtran}
\bibliography{privrefs}
%\begin{biography}
%{Jayakrishnan Unnikrishnan (S'06--M'10)} received the B.Tech. degree in electrical engineering from the Indian Institute of Technology Madras, Chennai, in 2005 and the M.S. and Ph.D. degrees in electrical and computer engineering from the University of Illinois at Urbana-Champaign in 2007 and 2010, respectively.
%
%From 2010 to 2014, he worked as a postdoctoral researcher at the School of Computer and Communication Sciences, \'{E}cole Polytechnique F\'{e}d\'{e}rale de Lausanne (EPFL), Lausanne, Switzerland. His current research interests include detection and estimation theory, signal processing, and information theory.
%
%Dr. Unnikrishnan is a recipient of the Vodafone Graduate Fellowship Award from the University of Illinois at Urbana-Champaign for 2007--2008 and the E.A. Reid Fellowship Award from the ECE department at the University of Illinois at Urbana-Champaign for 2010--2011.
%\end{biography} 

\end{document}